\documentclass{article}

\def\noheaderplainsetup{

\topmargin=0pt \headheight=0pt \headsep=0pt  \oddsidemargin=0pt \evensidemargin=0pt  \textheight=9.1truein \textwidth=6.5truein}   

\noheaderplainsetup

\usepackage{amsfonts}    

\begin{document}


\newcommand{\clsixteen}{\mbox{{\bf CL16}}}
\newcommand{\clseventeen}{\mbox{{\bf CL17}}}
\newcommand{\cleighteen}{\mbox{{\bf CL18}}}
\newcommand{\lp}[2]{\mbox{\bf Lp}^{#1}_{#2}}
\newcommand{\Rank}{\mbox{Rank}}
\newcommand{\rank}{\mbox{\scriptsize Rank}}
\newcommand{\Bigmlc}{\mbox{{\Large $\wedge$}}}
\newcommand{\Bigmld}{\mbox{{\Large $\vee$}}}
\newcommand{\bigmlc}{\mbox{{\large $\wedge$}}}
\newcommand{\bigmld}{\mbox{{\large $\vee$}}}
\newcommand{\bigleftbrace}{\mbox{{\large $\{$}}}
\newcommand{\bigrightbrace}{\mbox{{\large $\}$}}}
\newcommand{\Bigleftbrace}{\mbox{{\Large $\{$}}}
\newcommand{\Bigrightbrace}{\mbox{{\Large $\}$}}}
\newcommand{\emptyrun}{\langle\rangle}
\newcommand{\legal}[2]{\mbox{\bf Lr}^{#1}_{#2}} 
\newcommand{\win}[2]{\mbox{\bf Wn}^{#1}_{#2}} 
\newcommand{\seq}[1]{\langle #1 \rangle} 


\newcommand{\ade}{\mbox{\large $\sqcup$}}      
\newcommand{\ada}{\mbox{\large $\sqcap$}}      
\newcommand{\gneg}{\neg}                  
\newcommand{\mli}{\rightarrow}                     
\newcommand{\mld}{\vee}    
\newcommand{\mlc}{\wedge}  
\newcommand{\add}{\hspace{0pt}\sqcup}           
\newcommand{\adc}{\hspace{0pt}\sqcap}
\newcommand{\st}{\mbox{\raisebox{-0.05cm}{$\circ$}\hspace{-0.13cm}\raisebox{0.16cm}{\tiny $\mid$}\hspace{2pt}}}
\newcommand{\cost}{\mbox{\raisebox{0.12cm}{$\circ$}\hspace{-0.13cm}\raisebox{0.02cm}{\tiny $\mid$}\hspace{2pt}}}


\newtheorem{theoremm}{Theorem}[section]
\newtheorem{factt}[theoremm]{Fact}
\newtheorem{definitionn}[theoremm]{Definition}
\newtheorem{lemmaa}[theoremm]{Lemma}
\newtheorem{conventionn}[theoremm]{Convention}
\newtheorem{claimm}[theoremm]{Claim}
\newtheorem{examplee}[theoremm]{Example}
\newtheorem{corollaryy}[theoremm]{Corollary}
\newtheorem{remarkk}[theoremm]{Remark}

\newenvironment{definition}{\begin{definitionn} \em}{ \end{definitionn}}
\newenvironment{theorem}{\begin{theoremm}}{\end{theoremm}}
\newenvironment{lemma}{\begin{lemmaa}}{\end{lemmaa}}
\newenvironment{fact}{\begin{factt}}{\end{factt}}
\newenvironment{corollary}{\begin{corollaryy}}{\end{corollaryy}}
\newenvironment{example}{\begin{examplee}}{\end{examplee}}
\newenvironment{claim}{\begin{claimm}}{\end{claimm}}
\newenvironment{remark}{\begin{remarkk}}{\end{remarkk}}
\newenvironment{convention}{\begin{conventionn} \em}{\end{conventionn}}
\newenvironment{proof}{ {\bf Proof.} }{\  $\Box$ \vspace{.1in} }

\title{A propositional cirquent calculus for computability logic}
\author{Giorgi Japaridze
  \\  
 \\ Department of Computing Sciences, Villanova University, USA\\
 Email: giorgi.japaridze@villanova.edu\\
 URL: http://www.csc.villanova.edu/$^\sim$japaridz/
}
\date{}
\maketitle

\begin{abstract}  
{\em Cirquent calculus} is a proof system with inherent ability to account for sharing subcomponents in logical expressions. 
Within its framework, this article constructs an axiomatization $\cleighteen$ of the basic propositional fragment of {\em computability logic} --- the game-semantically conceived logic of computational resources and tasks. The nonlogical atoms of this fragment represent arbitrary so called static games,  and the connectives of its logical vocabulary are negation and the parallel and choice versions of   conjunction and disjunction. The main technical result of the article is a proof of the soundness and completeness of $\cleighteen$  with respect to the semantics of computability logic.  
\end{abstract}


\  


\section{Introduction}\label{intr}
{\em Computability logic}, or {\em CoL} for short, is a long haul venture for redeveloping  logic as a formal theory of computability, as opposed to its more traditional role  of a formal theory of truth (see \cite{iran} for an overview).  It starts by asking what kinds of mathematical objects ``computational problems''  are,  and finds that, in full generality, they can be most adequately understood as games played by a machine against its environment, with computability meaning existence of an (algorithmic) winning strategy for the machine. Then CoL tries to identify a collection of the most natural operations on games to constitute the operators of its logical vocabulary. Validity of a formula is understood as being ``always computable'', i.e., computable in virtue of the meanings of its logical operators regardless of how the non-logical atoms are interpreted.  The main and ongoing stage in developing CoL is finding adequate axiomatizations for ever more expressive fragments of this otherwise semantically construed and inordinately expressive logic. The present article makes a new step in this direction.  

Of the over three dozen logical operators introduced in CoL so far, the most basic ones are  {\em negation} \mbox{(``{\em not}\hspace{1pt}'')} $\neg$, {\em parallel conjunction} (``{\em pand}\hspace{1pt}'') $\mlc$, {\em parallel disjunction} (``{\em por}\hspace{1pt}'') $\mld$, {\em choice conjunction} (``{\em chand}\hspace{1pt}'') $\adc$ and  {\em choice disjunction} (``{\em chor}\hspace{1pt}'') $\add$.  Below is a brief informal explanation of the semantics of these operators as operations on games.   

The game $\gneg G$ is nothing but $G$ with the roles  of the two players interchanged. $G\mlc H$ and $G\mld H$ are games playing which means playing $G$ and $H$ in parallel; in order to be the winner in $G\mlc H$, the machine needs to win in both $G$ and $H$, while in the case of $G\mld H$ winning in just one component is sufficient.  $G\adc H$ is a game where the environment has to choose one of the two components, after which the game continues and the winner is determined according to the rules of the chosen component;  if the environment fails to make such a choice, it loses.  $G\add H$ is similar, with the difference that here it is the machine who has the privilege and duty to make an initial choice.  

Game operations with similar intuitive characterizations have been originally studied by Lorenzen \cite{Lor61}, Hintikka \cite{Hintikka73} and Blass \cite{Bla72,Bla92}. Hintikka interpreted the operators of classical logic in the style of our choice operators. So did Lorenzen for intuitionistic conjunction and disjunction, with the idea of parallel connectives (specifically disjunction) only implicitly present in his interpretation of the intuitionistic implication. Blass was the first to systematically differentiate between the parallel and choice sorts of operations and call attention to their resemblance with the multiplicative ($\mlc,\mld$) and additive ($\adc,\add$) connectives of Girard's \cite{Gir87} linear logic. The majority of other operators  studied within the framework of CoL, however, have no known analogs in the literature. 

Among the peculiarities of CoL is that it has two sorts of atoms: {\em general} atoms standing for any games,\footnote{More precisely, for any so called static games --- see the end of Section 2.} and {\em elementary} atoms standing for propositions, where the latter are understood as games with no moves, automatically won by the machine when true and lost when false. The fragments of CoL where only general atoms are present \cite{BauerLMCS, Cirq, Propint, Japfour, separating, taming1, taming2, Ver, qu, Xure1, Xure2} are said to be {\em general-base}, the fragments with only elementary atoms \cite{Japtocl1,Japtcs,lbcs,cl12,ifcolog,igpl21} are said to be {\em elementary-base}, and the fragments with both sorts of atoms  \cite{Japtocl2, Japtcs2, Japseq, Japtoggling, XuIGPL} are said to be  {\em mixed-base}. Technically, the fragment to which the present article is devoted is general-base; however, as will be shown in Section \ref{from}, there is a straightforward embedding of the corresponding mixed-base fragment into it, which allows us to say that the present fragment is in fact mixed-base.

No axiomatizations of CoL or its nontrivial fragments had been found within the framework of traditional proof calculi, and it was conjectured \cite{Bla92,Cirq} that such an axiomatization was impossible to achieve in principle even for just the $\{\neg,\mlc,\mld\}$-subfragment. This conjecture was later proven to be true by Das and Strassburger in  \cite{Anupam}. 
To break the ice, \cite{Cirq} introduced the new sort of a proof calculus termed {\em cirquent calculus}, in which a sound and complete axiomatization of the general-base $\{\neg,\mlc,\mld\}$-fragment of CoL was constructed; this result was later lifted to the mixed-base level in \cite{XuIGPL}. Instead of tree-like objects such as formulas, sequents, hypersequents \cite{Avron} or deep-inference structures \cite{Brunnler01,Gug07,Guglielmi01}, cirquent calculus deals with circuit-style constructs dubbed {\em cirquents}. Cirquents come in a variety of forms and sometimes, as in the present work or in \cite{ifcolog,igpl21,XuIf,XuLast}, they are written textually (thus resembling formulas) rather than graphically, but their essence and main distinguishing feature remain the same: these are syntactic constructs explicitly allowing {\em sharing} of components between different subcomponents. Ordinary formulas of CoL are nothing but special cases of cirquents --- degenerate cirquents where nothing is shared. 

Sharing, itself, also takes various forms, such as, for instance, two $\mld$-gates sharing a child, or two $\add$-gates sharing the left-or-right choice associated with them without otherwise sharing descendants.  Most cirquent calculus systems studied so far \cite{BauerLMCS, Cirq, Japdeep, taming1, taming2, XuIGPL} only incorporate the first sort of sharing. The idea of the second sort of sharing, termed {\em clustering}, was introduced and motivated in \cite{lmcs}. Among the potential benefits of it outlined in \cite{lmcs} was offering new perspectives on independence-friendly logic \cite{HS97}. Later work by Wenyan Xu \cite{XuIf,XuLast} made significant progress towards materializing such a potential. Another benefit offered by clustering was materialized in \cite{ifcolog}, where a cirquent calculus axiomatization  $\clsixteen$ of the elementary-base $\{\neg,\mlc,\mld,\adc,\add\}$-fragment of CoL was found. No axiomatizations of any $\adc,\add$-containing fragments of CoL had been known before that (other than the brute-force  constructions of \cite{Japtocl1,Japtocl2,Japtcs,Japtcs2,Japseq,Japtoggling}, with their deduction mechanisms more resembling games than logical calculi). Generalizing from formulas to cirquents with clustering thus offers not only greater expressiveness, but also makes the otherwise unaxiomatizable CoL or certain fragments of it  amenable to being tamed as logical calculi. 

The present contribution introduces the cirquent calculus system $\cleighteen$ for the general-base---and, as noted, essentially also mixed-base---$\{\neg,\mlc,\mld,\adc,\add\}$-fragment of CoL  and proves its soundness and completeness. This system extends and strengthens $\clsixteen$ in a substantial way, lifting it from the (relatively simple) elementary-base level to the  (usually much harder to tame) general-base level and extending the idea of clustering from logical operators to nonlogical atoms.

\section{Games and interactive machines}\label{sgames}
As  noted, CoL understands computational problems as games played between two players, called  the {\em machine} and   the  {\em environment}, symbolically named $\top$ and $\bot$, respectively. $\top$ is a deterministic mechanical device only capable of following algorithmic strategies, whereas there are no restrictions on the behavior of $\bot$.  
$\wp$ is always  a variable ranging over $\{\top,\bot\}$. \(\overline{\wp}\) means $\wp$'s adversary, i.e., the player that is not $\wp$.

A {\bf move} is a finite string over the standard keyboard alphabet. 
A {\bf labmove} (abbreviates ``labeled move'') is a move prefixed with $\top$ or $\bot$, with such a prefix ({\bf label}) indicating which player has made the move. 
A {\bf run} is a (finite or infinite) sequence of labmoves, and a 
{\bf position} is a finite run.
Runs will be often delimited by ``$\langle$'' and ``$\rangle$'', as in $\seq{\Gamma}$, $\seq{\Gamma,\top\alpha,\Theta}$ etc.  
 
A set $S$ of runs is said to be {\bf prefix-closed} iff, whenever a run is in $ S$, so are all of its initial segments. And a set $S$ of runs is {\bf limit-closed} iff, for any infinite run $\Gamma$,  if all finite initial segments of $\Gamma$ are in  $S$, then so is $\Gamma$ itself.    

\begin{definition}\label{game}
A {\bf game} is a pair  $G=(\legal{G}{},\win{G}{})$, where:

1.   $\legal{G}{}$  is a nonempty, prefix- and limit-closed set of runs. 

2.      $\win{G}{}$ is a mapping from $\legal{G}{}$ to $\{\top,\bot\}$. 
\end{definition}

Note that, since $\legal{G}{}$ is required to be nonempty and prefix-closed, it always contains the {\bf empty position/run} $\seq{}$.  

Intuitively, $\legal{G}{}$ is the set of  {\bf legal runs} of $G$ and, for each such run, $\win{G}{}$ tells us which of the two players $\wp\in\{\top,\bot\}$ has   won  the run. Correspondingly, when $\win{G}{}\seq{\Gamma}=\wp$, we say that $\Gamma$ is a  {\bf $\wp$-won}, or {\bf won by} $\wp$,  run of $G$.
In all cases, we shall say ``{\bf illegal}'' for ``not legal'' and  ``{\bf lost}'' for ``not won''. 
A {\bf legal move} by  $\wp$  in a position $\Theta$ is a move $\alpha$ such that $\seq{\Theta,\wp\alpha}\in \legal{G}{}$. We say that a run $\Gamma$ 
is {\bf $\wp$-legal}  iff    either $\Gamma$ is legal, or else, where $\Theta$ is the shortest illegal initial segment of $\Gamma$, the last move 
of $\Theta$ is $\overline{\wp}$-labeled. Intuitively, such a $\Gamma$ is a run where $\wp$ has not made any illegal moves unless its 
adversary $\overline{\wp}$ has done so first. 

The following definition extends the 
meaning of the word ``won''  from legal runs to all runs by stipulating that an illegal run is always lost by the player that has made the first illegal move:

\begin{definition}\label{newdef}
For a game $G$, run $\Gamma$ and player $\wp$, we say that $\Gamma$ is a {\bf $\wp$-won} (or won by $\wp$) run of $G$ iff $\Gamma$ is 

1. either  a legal run of $G$ with $\win{G}{}\seq{\Gamma}=\wp$, 

2. or a $\overline{\wp}$-illegal run of $G$. 

\end{definition}

As an example, consider the problem of deciding a unary predicate $p(x)$. It can be understood as the game whose every legal position is either $\seq{ \bot m,\top \mbox{yes}}$ or $\seq{ \bot m,\top \mbox{no}}$ or $\seq{\bot m}$ or $\seq{}$, where $m$ is a number written in (for instance) decimal notation;
the empty run $\seq{}$ of this game is won by $\top$, a length-1 run $\seq{\bot m}$ is won by $\bot$, and a length-2 run 
$\seq{ \bot m,\top \mbox{yes}}$ (resp. $\seq{ \bot m,\top \mbox{no}}$) is won by $\top$ iff $p(m)$ is true 
(resp. false). The intuitions that determine  these arrangements are as follows. A move $\bot m$ is the ``input'', making 
which amounts to asking the machine ``is $p(m)$ true?''. The move $\top \mbox{yes}$ or  $\top \mbox{no}$ is the ``output'', amounting to correspondingly answering ``yes'' or ``no''.   The machine wins the run $\seq{\bot m,\top \mbox{yes}}$ iff $p(m)$ is (indeed) true, and wins the run $\seq{\bot m,\top \mbox{no}}$ iff $p(m)$ is (indeed) not true.   The run $\seq{\bot m}$ corresponds to the situation where there was an input/question but the machine failed to generate an  output/answer,  so the machine loses. And the run $\seq{}$ corresponds to the situation where there was no input either, so the machine has nothing to answer for and it therefore does not lose---i.e., it wins---the game. 

A game is said to be {\bf elementary} iff it has no legal runs other than the (always legal) empty run $\emptyrun$. That is, an elementary game is a game without any (legal) moves,  automatically won or lost. There are exactly two  such games, for which we use the same symbols $\top$ and $\bot$ as for the two players: the game $\top$ automatically won by player $\top$, and the game $\bot$ automatically won by player $\bot$.\footnote{Precisely, we have $\win{\top}{}\emptyrun=\top$ and $\win{\bot}{}\emptyrun=\bot$.} Computability logic is a conservative extension of classical logic, understanding classical propositions as elementary games. And, just like classical logic, it sees no difference between any two true propositions such as ``$2+2=4$'' and ``{\em Snow is white}'', identifying both with the elementary game $\top$; similarly, it treats false propositions such as ``$2+2=5$'' or ``{\em Snow is black}'' as the elementary game $\bot$. 

Algorithmic game-playing strategies we understand as the interactive versions of Turing machines named {\bf HPM}s (``Hard-Play Machines''). An HPM  is a Turing machine with the additional capability of making moves. The adversary can also move at any time, and such moves are the only nondeterministic events from the machine's perspective. Along with the ordinary  read/write {\em work tape},\footnote{In computational-complexity-sensitive treatments, an HPM is allowed to have any (fixed) number of work tapes.} the machine also has an additional  tape called  the  {\em run tape}. The latter, at any time,  spells the ``current position'' of the play. The role of this tape is to make the interaction history fully visible to the machine.  It is read-only, and its content is automatically updated every time either player makes a move. A more detailed description of  the HPM model, if necessary, can be found in \cite{Jap03}. 

In these terms,  a  {\bf solution} ($\top$'s winning strategy) for a given  game $A$ is understood as an HPM $\cal M$ such that, no matter how the environment acts during its interaction with $\cal M$ (what moves it makes and when), the run incrementally spelled on the run tape is  $\top$-won. When this is the case, we write ${\cal M}\models A$ and say that ${\cal M}$ {\bf wins}, or {\bf solves}, $A$, and that $A$ is a {\bf computable} game.   

There is no need to define $\bot$'s strategies, because all possible behaviors by $\bot$ are accounted for by the different possible nondeterministic updates of the run tape of an HPM. 

In the above outline, we described HPMs in a relaxed fashion, without being specific about technical details such as, for instance, how, exactly, moves are made by the machine, how many moves either player can make at once, what happens if both players attempt to move ``simultaneously'', etc. As it happens, all reasonable arrangements yield the same class of winnable games as long as we consider a certain natural subclass of games called {\em static}. 
Intuitively, these are games where the relative speeds of the players are irrelevant because, Andreas Blass put it in his review of \cite{Jap03}, ``it never hurts a player to postpone making moves''. Below comes a formal definition of this concept.

For either player $\wp$, we say that a run $\Upsilon$ is a {\bf $\wp$-delay} of a run $\Gamma$ iff:\vspace{-5pt}
\begin{itemize}
\item for both players $\wp'\in\{\top,\bot\}$, the subsequence of $\wp'$-labeled moves of $\Upsilon$ is the same as that of $\Gamma$, and
\item for any $n,k\geq 1$, if the $n$th $\wp$-labeled move is made later than (is to the right of) the $k$th $\overline{\wp}$-labeled move in $\Gamma$, then so is it in $\Upsilon$.\vspace{-5pt}
\end{itemize}
\noindent The above conditions mean that in  $\Upsilon$  each player has made the same sequence of moves as in $\Gamma$, only, in $\Upsilon$, $\wp$ might have been acting with some delay. Example: $\seq{\bot\beta_1,\top \alpha_1,\bot\beta_2, \top\alpha_2,\bot\beta_3,\top\alpha_3}$ is a $\top$-delay of  $\seq{\top \alpha_1,\bot\beta_1,\bot\beta_2,\top\alpha_2,\top\alpha_3,\bot\beta_3}$.

\begin{definition}
\label{sgd}
We say that a game  $A$ is {\bf static} iff, for either player $\wp$, whenever a run $\Upsilon$ is a $\wp$-delay of 
a run $\Gamma$, we have:\vspace{-2pt}
\begin{enumerate}
\item if $\Gamma$ is a $\wp$-legal run of $A$, then so is $\Upsilon$;
\item if $\Gamma$ is a $\wp$-won run of $A$, then so is $\Upsilon$.
\end{enumerate}\end{definition}

All games that we shall see in this paper are static.
Dealing only with such games, which makes timing technicalities fully irrelevant, allows us to describe and analyze strategies (HPMs) in a relaxed fashion. For instance, imagine an HPM $\cal N$ --- a purported solution of a static game $G$ --- works by simulating and mimicking the work and actions of another HPM $\cal M$ in the scenario where $\cal M$'s imaginary adversary acts in the same way as $\cal N$'s own adversary. Due to the simulation overhead, $\cal N$ will generally be slower than $\cal M$ in responding to its adversary's moves.  Yet, we may safely assume/pretend/imagine that the speeds of the two machines do not differ; that is in the sense that if the imaginary (fast) version of $\cal N$ wins $G$, then so does the actual (slow) $\cal N$.      In what follows, we will often implicitly rely on this observation. 

Throughout this paper, for a run $\Gamma$, we let $\overline{\Gamma}$ mean the result of  simultaneously changing the label  $\top$ to $\bot$ and  vice versa in each labmove of $\Gamma$.   

\begin{definition} The {\bf negation} of a game $A$ is defined as the game $\neg A$ such that:  
\begin{itemize}
  \item $\legal{\neg A}{}=\{\overline{\Gamma}\ |\ \Gamma\in\legal{A}{}\}$;
  \item $\win{\neg A}{}\seq{\Gamma}=\top$ iff $\win{A}{}\seq{\overline{\Gamma}}=\bot$ (for any legal run $\Gamma$
of $\neg A$).
\end{itemize}
\end{definition} 

Intuitively, as noted earlier, $\neg A$ is $A$ with the roles of the two players interchanged: $\top$'s (legal) moves and wins become $\bot$'s moves and wins, and vice versa.  For instance, if  \mbox{\em Chess} is the game of chess from the point of view of the white player, then $\neg\mbox{\em Chess}$ is the same game but as seen by the black player.
Whenever $A$ is a static game, so is $\neg A$, as the operation of negation --- as well as all other game operations studied in CoL --- is known \cite{Jap03} to preserve the static property. 
  
\section{Syntax}\label{Syntax}

We fix two groups of syntactic objects used in the language of $\cleighteen$: (general game) {\bf letters} $Letter_0$, $Letter_1$, $Letter_2$,\ldots  and {\bf clusters} $0,1,2,\ldots$.  As metavariables, we will be using $P,Q,R,S,\ldots$ for letters and  $a,b,c,d,\ldots$ for clusters. A (nonlogical) {\bf atom} is a pair consisting of a letter $P$ and a cluster $a$, written as $P^a$. A {\bf literal} is $P^a$ ({\bf positive} literal) or $\neg P^a$ ({\bf negative} literal), where $P^a$ is an atom. Such a literal or atom is said to be {\bf $P$-lettered} and {\bf $a$-clustered}; we can also say that $P$ is the {\bf letter of} and $a$ the {\bf cluster of} such a literal or atom.  A {\bf choice connective} is a pair $O^a$, where $O\in\{\adc,\add\}$ and $a$ is a cluster, said to be the {\bf cluster of} the connective.

To get some preliminary insights into the intended purpose of clusters, here we make an informal digression before continuing with the definition of the syntax of   $\cleighteen$. 
While the letters of  literals are  meant to represent games, the role of the clusters of literals is to indicate whether different occurrences of same-lettered literals in an expression are different copies of the corresponding  game or the same copy but written at different places. Clustering allows us to write cirquents in a linear (formula-style) fashion; otherwise, if we used graphical (circuit-style) representations of cirquents, there would be no need for clusters. 
For instance, consider the game $G$ represented by the following circuit-style construct:

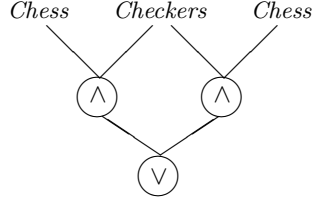
\begin{center} 

\begin{picture}(120,80)

\put(0,65){\small{\em Chess}}
\put(40,65){\small{\em Checkers}}
\put(92,65){\small{\em Chess}}

\put(35,42){\line(1,1){20}}
\put(35,42){\line(-1,1){20}}
\put(82,42){\line(-1,1){20}}
\put(82,42){\line(1,1){20}}
\put(34,35){\circle{15}}
\put(81,35){\circle{15}}
\put(57,5){\circle{15}}
\put(31,33){$\mlc$}
\put(78,33){$\mlc$}
\put(54,2){$\mld$}
\put(58,13){\line(3,2){22}}
\put(58,13){\line(-3,2){22}}

\end{picture}
\end{center}

\noindent Such a $G$ is a parallel/simultaneous play on three boards, where {\em Chess} is played on boards \#1 and \#3, and {\em Checkers} played on board \#2. In order to win, $\top$ needs to win either on both  boards \#1 and ($\mlc$) \#2, or ($\mld$) on both boards \#2 and \#3 (or on all three boards). Note that, even though the plays on boards \#1 and \#3 should follow the same rules (those of {\em Chess}), the two plays may run in different ways due to different moves made in them, which makes it possible that one is eventually won while the other is lost. Anyway, how could $G$ be written in a linear fashion? The expression 
\begin{equation}\label{ttp}
(\mbox{\em Chess}\mlc\mbox{\em Checkers})\mld(\mbox{\em Checkers}\mlc\mbox{\em Chess})\end{equation}
would not be adequate, as it would suggest a play on four rather than three boards. This is where clusters (of literals) come in. To indicate that the two occurrences of {\em Checkers} in (\ref{ttp}) are in fact the same copy of {\em Checkers} while the two occurrences of {\em Chess} are two different copies of {\em Chess}, we can assign a same cluster $b$
to the former and two different clusters $a,c$ to the latter. This way, $G$ can be adequately written as 
\[(\mbox{\em Chess}^a\mlc\mbox{\em Checkers}^b)\mld(\mbox{\em Checkers}^b\mlc\mbox{\em Chess}^c),\]
showing that {\em Checkers} is {\em shared} between the two conjunctive subcomponents while {\em Chess} is not. 

As for connective clustering, it is associated with a different sort of sharing: sharing of choices.  Imagine two subcomponents $A_0\add A_1$ and $B_0\add B_1$ in some complex expression/game.
From what was said in Section \ref{intr}, we know that $A_0\add A_1$ allows $\top$ to choose between $A_0$ and $A_1$, and similarly for $B_0\add B_1$. If the choices in the two components are inependent  from each other, we would write those components as $A_0\add^a A_1$ and $B_0\add^b B_1$ for some clusters $a\not= b$. But another possibility is that the two left-or-right  choices are required to be the same. That is, choosing $A_0$ (resp. $A_1$) in  $A_0\add A_1$ automatically results in choosing $B_0$ (resp. $B_1$) in $B_0\add B_1$. This arrangement will be accounted for by writing the two $\add$-disjunctions as $A_0\add^c A_1$ and $B_0\add^c B_1$ for some---the same for both $\add$s---cluster $c$. 

It should be pointed out that two  connectives having the same cluster indicates a shared choice between them only if both connectives are of the same type, i.e., both are $\adc$ or both are $\add$. So, there is no semantical difference whatsoever between, for instance, the expressions $(P\add^1 Q)\mlc(R\adc^1 S)$ and $(P\add^1 Q)\mlc(R\adc^2 S)$. Similarly,  two  literals having the same cluster indicates a shared run of a game only if both literals have the same letter and both are positive or both are negative. So,  when $Q\not=P$, there is no semantical difference between the expressions $P^1\mld Q^1$ and $P^1\mld Q^2$, or between  $P^1\mld \neg P^1$ and $P^1\mld \neg P^2$. Needless to say that a connective and a literal having the same cluster does not signify any connection between the two.

\begin{definition}\label{cirquent}
A {\bf cirquent} is defined inductively as follows:

\begin{itemize}
  \item $\top$ and $\bot$ are cirquents. 
  \item Each (positive or negative) literal is a cirquent. 
\item If $A$ and $B$ are cirquents,  then $(A)\mld (B)$ and $(A)\mlc (B)$ are cirquents.  Such cirquents are respectively said to be {\bf $\mld$-rooted} and {\bf $\mlc$-rooted}. 
  \item If $A$ and $B$ are cirquents and   $c$ is a cluster,  then $(A)\adc^c(B)$  and $(A)\add^c(B)$ are  cirquents.  Such a cirquent is said to be {\bf $\adc^c$-rooted} or {\bf $\add^c$-rooted}, respectively 
When we simply say {\bf $\adc$-rooted} or {\bf $\add$-rooted}, it means  $\adc^c$-rooted or $\add^c$-rooted for some (whatever) cluster $c$.
\end{itemize}
\end{definition}

For readability, we usually omit parentheses in cirquents as long as this causes no ambiguity. 
Sometimes we may write an expression such as $A_1\mld\ldots \mld A_n$, where $n$ is a (possibly unspecified) positive integer. This is to be understood as just $A_1$ when $n=1$, and as any order-respecting $\mld$-combination of the cirquents $A_1,\ldots,A_n$ when $n\geq 2$.  ``Order-respecting'' in the sense that $A_1$ is the leftmost item of the combination, then
 comes $A_2$, then $A_3$, etc. Similarly for $A_1 \mlc\ldots\mlc A_n$. So, for instance, both $(A\mlc B)\mlc C$ and $A\mlc(B\mlc C)$—
 and no other cirquent—can be written as $A\mlc B \mlc C$.

\section{Semantics}

As a matter of notational convention, we agree that, for a run $\Gamma$ and literal $L$,  $\Gamma^{L.}$ means the run obtained from $\Gamma$ by first deleting all labmoves except those that look like $\wp L.\alpha$ for some string $\alpha$ and $\wp\in\{\top,\bot\}$, and then further replacing each such labmove  $\wp L.\alpha$ with just $\wp\alpha$. Intuitively, a move of the form $ L.\alpha$  is the move $\alpha$ made in all  copies of the game represented by $L$, and $\Gamma^{L.}$ is the (common) subrun of $\Gamma$ that took place in those copies. Further, an {\bf in-literal labmove} is a labmove of the form $L.\alpha$, where $L$ is a literal.
And a {\bf choice labmove} is either $\bot\adc^c\hspace{-2pt}.i$ or  $\top\add^c\hspace{-2pt}.i$ for some cluster $c$ and  $i\in\{0,1\}$.

An {\bf interpretation}  is a function $^*$ that assigns a static game $P^*$ to  each letter  $P$. Given an interpretation $^*$, we immediately extend it to literals  by stipulating that, for every letter $P$ and cluster $c$, $(P^c)^*= P^*$ and $(\neg P^c)^*=\neg(P^*)$.

\begin{definition}
\label{clsem}
Every cirquent $C$ and interpretation $^*$ induces the unique game $C^*$, to which we refer as ``$C$ {\bf under} $^*$'', defined as follows.

\begin{description}
  \item[$\legal{C^*}{}$] A run $\Gamma$ is a legal run of $C^*$---i.e., $\Gamma\in\legal{C^*}{}$---iff every labmove of $\Gamma$ is either an in-literal or a choice one, and  the following two conditions are satisfied:  
\begin{enumerate}
  \item For every literal $L$,   $\Gamma^{L.}$ is a legal run of $L^*$. 
  \item For every cluster $c$,  $\Gamma$ contains at most one  labmove of the form  $\bot \adc^c\hspace{-2pt}.i$ (resp. $\top \add^c\hspace{-2pt}.i$), where $i\in\{0,1\}$.  When such a  labmove is present,  
any (sub)cirquent of the form $A_0\adc^c A_1$  (resp. $A_0\add^c A_1$) is said to be {\bf $\Gamma$-resolved}, and $A_i$ is said to be its {\bf $\Gamma$-resolvent};  otherwise we say that $A_0\adc^c A_1$  (resp. $A_0\add^c A_1$) is {\bf $\Gamma$-unresolved}.  Here the prefix ``$\Gamma$-'' can be omitted before the word ``(un)resolved'' when clear from the context.   The terms ``resolved'' and ``unresolved'', with the same meanings, can as well be used in reference to just the connective $\adc^c$ (resp. $\add^c$).  

\end{enumerate}

\item[$\win{C^*}{}$]  A legal run $\Gamma$ of $C^*$ is $\top$-won---i.e.,  $\win{C^{*}}{}\seq{\Gamma}=\top$---iff one of the following conditions is satisfied: 
\begin{enumerate}
\item  $C$ is $\top$.
\item $C$  is a literal  $L$, and $\win{L^*}{}\seq{\Gamma^{L.}}=\top$.  
\item $C$ has the form $A_0\mlc A_1$ (resp. $A_0\mld A_1$) and, for both (resp. at least one) $i\in\{0,1\}$,  $\win{A_{i}^{*}}{}\seq{\Gamma}=\top$.
\item $C$ is a $\Gamma$-resolved $\adc^c$- or $\add^c$-rooted (for whatever cluster $c$) cirquent  and,  where $B$ is the resolvent, $\win{B^{*}}{}\seq{\Gamma}=\top$.   
\item $C$ is a $\Gamma$-unresolved $\adc^c$-rooted  (for whatever cluster $c$) cirquent. 
\end{enumerate}
\end{description}
\end{definition}

When an interpretation $^*$ is fixed in the context or is irrelevant, by abuse of notation we may omit explicit references to it, and identify a cirquent $C$ with the game $C^*$. For instance, we may say that the machine wins $C$ instead of saying that the machine wins $C$ under $^*$. Notice also that $\legal{C^*}{}$ only depends on $^*$ but not on $C$, which, in similar contexts, allows us to just say ``{\bf $^*$-legal run}" (or just ``legal run" if $^*$ is clear from the context) instead of ``legal run of $C^*$''. 

The intuitive meaning of a move $\adc^c\hspace{-2pt}.i$ (resp. $\add^c\hspace{-2pt}.i$), which can be legally made only by player $\bot$ (resp. $\top$) and made only once, is choosing the component $A_i$ in every subcirquent of $C$ of the form $A_0 \adc^c A_1$ (resp. $A_0 \add^c A_1$). The effect of such a move is that, when determining the outcome of the play, $A_0 \adc^c A_1$ (resp. $A_0 \add^c A_1$) will be treated as if it was simply $A_i$. And if such a choice move is absent, then $A_0 \adc^c A_1$ (resp. $A_0 \add^c A_1$) will be treated as if it was simply $\top$ (resp. $\bot$); this is so because making a choice between $A_0$ and $A_1$ is not only a privilege, but also an obligation of the corresponding player.

It should be noted that the relaxed, cirquent-independent design choice for the $\legal{}{}$ component of our present games deems certain meaningless yet ``harmless'' moves legal, such as moves in non-existent choice connectives or non-existent literals. Excluding such moves would result in more elegant definitions without otherwise affecting winnability of games, but create significant technical inconveniences in our subsequent proofs. 

\begin{fact}\label{stg}
For any cirquent $C$ and interpretation $^*$, the game $C^*$ is static.
 \end{fact}

\begin{proof} Consider arbitrary cirquent $C$,  interpretation $^*$, player $\wp\in\{\top,\bot\}$ and two runs $\Gamma,\Upsilon$, where $\Upsilon$ is a $\wp$-delay of $\Gamma$. In what follows, whenever we say ``legal", ``won", etc., it is to be understood in the context of $C^*$. 

We first need to see that condition 1 of Definition \ref{sgd} is satisfied. We do this by verifying its contrapositive.   Assume the run $\Upsilon$ is $\wp$-illegal. We want to show that then so is $\Gamma$, which will be done by induction on the length of the  shortest 
illegal initial segment of $\Upsilon$. 
Let $\seq{\Psi,\wp\alpha}$ be such a segment. Let $\seq{\Phi,\wp\alpha}$ be the shortest initial segment of $\Gamma$ containing all the $\wp$-labeled moves\footnote{In this context, we are talking about occurrences of labmoves rather than labmoves as such, with each occurrence counting. So, a more accurate phrase here would be to ``... as many $\wp$-labeled moves as...''.}    $\seq{\Psi,\wp\alpha}$. If the position $\Phi$ is $\wp$-illegal, 
then so is the run $\Gamma$ and we are done. Therefore, for the rest of our reasoning, assume that 
\begin{equation}\label{654}
\mbox{\em $\Phi$ is $\wp$-legal.}
\end{equation}

Let $\Theta$ be the sequence of those $\overline\wp$-labeled moves of $\Psi$ that are not in $\Phi$. Obviously
\begin{equation}\label{e141}
\mbox{\em  $\seq{\Psi,\wp\alpha}$ is a $\wp$-delay of $\seq{\Phi,\wp\alpha,\Theta}$.}
\end{equation}
We also claim that
\begin{equation}\label{e142}
\mbox{\em $\Phi$ is legal.}
\end{equation}
Indeed, suppose this was not the case. Then, by (\ref{654}), $\Phi$ should be $\overline\wp$-illegal. This would make $\Gamma$  a $\overline\wp$-illegal run of $C^*$ with $\Phi$ as an illegal initial segment which is shorter than $\seq{\Psi,\wp\alpha}$. Then, by the induction hypothesis, any run for which $\Gamma$ is a $\overline\wp$-delay, would be $\overline\wp$-illegal. But, as observed in Lemma 4.6 of \cite{Jap03}, the fact that $\Upsilon$ is a $\wp$-delay of $\Gamma$ implies that $\Gamma$ is a $\overline\wp$-delay of $\Upsilon$. So, $\Upsilon$ would be $\overline\wp$-illegal, which is a contradiction because, according to our assumptions, $\Upsilon$ is $\wp$-illegal.

We are continuing our proof. There are three possible reasons to why $\seq{\Psi,\wp\alpha}$ is  illegal (while $\Psi$ being legal):

{\em Reason 1}: $\wp\alpha$ is neither an in-literal nor a choice labmove. Then, in view of (\ref{e142}), $\seq{\Phi,\wp\alpha}$ is $\wp$-illegal. But then, as $\seq{\Phi,\wp\alpha}$ happens to be an initial segment of $\Gamma$, the latter  is also $\wp$-illegal, as desired.

{\em Reason 2}: $\wp\alpha$ is a choice labmove $\wp O^c.i$ (where $O=\adc$ if $\wp=\bot$ and $O=\add$ if $\wp=\top$) such that $\Psi$ already contains a labmove $\wp O^c.j$. But  then, again, in view of (\ref{e142}), the same reason makes $\seq{\Phi,\wp\alpha}$ snd thus $\Gamma$ also $\wp$-illegal, as desired.

{\em Reason 3}: $\wp\alpha$ is an in-literal labmove $\wp L.\beta$ such that  $\seq{\Psi,\wp\alpha}^{L.}$ is a ($\wp$-)illegal position of $L^*$. It is known (cf. \cite{Jap03}) that the operation of negation preserves the static property of games. Hence, the game $L^*$ is static regardless of whether the literal $L$ is positive or negative. Notice that (\ref{e141}) implies that  $\seq{\Psi,\wp\alpha}^{L.}$ is a $\wp$-delay of $\seq{\Phi,\wp\alpha,\Theta}^{L.}$.  So, as $L^*$ is a static game,  $\seq{\Phi,\wp\alpha,\Theta}^{L.}$ is a $\wp$-illegal position of $L^*$ (for otherwise, by condition 1 of Definition \ref{sgd}, $\seq{\Psi,\wp\alpha}^{L.}$ would be a $\wp$-legal position of $L^*$). But the $\Theta$ part of  $\seq{\Phi,\wp\alpha,\Theta}$ contains no $\wp$-labeled moves, and therefore $\seq{\Phi,\wp\alpha}^{L.}$  is a $\wp$-illegal position of $L^*$. This, in view of (\ref{e142}), implies that  $\seq{\Phi,\wp\alpha}$  is a $\wp$-illegal run of $C^*$. Then so is $\Gamma$, as $\seq{\Phi,\wp\alpha}$ is an initial segment of it. This completes our proof of the fact that condition 1 of Definition \ref{sgd} is satisfied.

For condition 2 of Definition \ref{sgd}, assume $\Gamma$ is a $\wp$-won run (of $C^*$). Our goal is to show that $\Upsilon$ is also $\wp$-won. If $\Upsilon$ is $\overline{\wp}$-illegal, then it is won by $\wp$ and we are done. So, assume that
$\Upsilon$ is $\overline{\wp}$-legal. According to the earlier mentioned Lemma 4.6 of \cite{Jap03}, if $\Upsilon$ is a $\wp$-delay of $\Gamma$, then $\Gamma$ is a $\overline{\wp}$-delay of $\Upsilon$. So, by the already verified condition 1 of Definition \ref{sgd} and the fact that $\Upsilon$ is $\overline{\wp}$-legal, we find that $\Gamma$ is $\overline{\wp}$-legal. 
$\Gamma$ is also  $\wp$-legal, because otherwise it would not be won by $\wp$. Consequently,  by condition 1 of Definition \ref{sgd}, $\Upsilon$ is $\wp$-legal as well.  To summarize, $\Gamma$ and $\Upsilon$ are both $\wp$- and $\overline{\wp}$-legal. Thus, we have narrowed down our considerations to the case where both $\Gamma$ and $\Upsilon$ are (simply) legal. Keeping this in mind, 
we now separately consider two cases, depending on whether $\wp$ is $\top$ or $\bot$.

{\em Case of $\wp=\top$.} According to our assumptions, $\Gamma$ is a $\top$-won run of $C^*$, and $\Upsilon$ is a $\top$-delay of $\Gamma$.
We want to show that (not only $\Gamma$ but also) $\Upsilon$ satisfies the five conditions of the $\win{C^*}{}$ clause of Definition \ref{clsem}. This can be done by induction on the complexity of $C$. Condition \#1 of the five conditions is trivial, and conditions \#3, \#4 and \#5 (inductive step) are also straightforward, taking into account that $\Gamma$ and $\Upsilon$ contain exactly the same choice labmoves. So, let us just look at condition \#2. 
Assume $C$ is a literal $L$.  As $\Gamma$ is a $\top$-won run of $C^*$,   $\Gamma^{L.}$ is   a $\top$-won run of $L^*$.
 But $\Upsilon^{L.}$ is a $\top$-delay of $\Gamma^{L.}$.
Therefore, since (as observed earlier) the game $L^*$  is static, 
$\Upsilon^{L.}$ is a  
$\top$-won run of $L^*$, as desired.

{\em Case of $\wp=\bot$.} According to our assumptions, $\Gamma$ is a $\bot$-won run of $C^*$, and $\Upsilon$ is a $\bot$-delay of $\Gamma$. We want to show that $\Upsilon$ is $\bot$-won. Deny this. Then $\Upsilon$ is $\top$-won. Also, according to the earlier mentioned Lemma 4.6 of \cite{Jap03}, $\Gamma$ is a $\top$-delay of $\Upsilon$. So, by the already verified Case  of $\wp=\top$, $\Gamma$ is $\top$-won, which contradicts our assumption that $\Gamma$ is $\bot$-won and hence proves that   $\Upsilon$ is $\bot$-won, as desired.
\end{proof}

\begin{definition} \label{krisha}
  A {\bf logical}  {\bf solution} of a cirquent $C$ is an HPM $\cal H$ such that, for any interpretation $^*$, $\cal H$ is a solution (cf. Section \ref{sgames})  of $ C ^*$. We say that $C$ is  {\bf valid} if it has a logical solution.
\end{definition}

\section{Axiomatics}\label{lll}

An {\bf inference rule} is a set $\cal R$ of pairs $\vec{A}\leadsto B$, where $\vec{A}$ is either a  cirquent or a pair of cirquents, called the {\bf premise(s)}, and $B$ is a cirquent, called the {\bf conclusion}. 
When $\vec{A}\leadsto B$ is in $\cal R$, we say that $B$ {\bf follows from} $\vec{A}$ by rule $\cal R$. 
   
We will be schematically representing all single-premise  inference rules in the form 
\[
X[B_{1},\ldots,B_{k}]
\ \leadsto \ X[A_1,\ldots, A_k],\]  
where $X[A_1,\ldots, A_k]$ is the conclusion together with subcirquents $A_1,\ldots,A_k$, and  $X[B_{1},\ldots,B_{k}]$ is the premise cirquent resulting from $X[A_1,\ldots, A_k]$ by replacing all occurrences of the subcirquents $A_1,\ldots,A_k$ with $B_{1},\ldots,B_{k}$, respectively.  Here (as in the Cleansing rules below) some $A_i$, in turn, can be written as $Y[C_1,\ldots,C_m]$ and the corresponding $B_i$ as $Y[D_1,\ldots,D_m]$. In such a case, again,  $Y[D_1,\ldots,D_m]$ should be understood as the result of replacing, in $Y[C_1,\ldots,C_m]$, all occurrences of the subcirquents $C_1,\ldots,C_m$ with $D_{1},\ldots,D_{m}$, respectively.

Given a cirquent $C$, a {\bf pseudoelementary} letter of $C$ is a letter $Q$ occurring in $C$  such that $C$ contains at most one positive and at most one negative  (but possibly with multiple occurrences)  $Q$-lettered literal.  That is, there is at most one cluster $a$ such that $C$ contains the literal $Q^a$ and at most one cluster $b$ such that $C$ contains the literal $\neg Q^b$. A {\bf non-pseudoelementary} letter of $C$ is a letter that occurs in $C$ but is not a pseudoelementary letter of $Q$. Further, a (non-)pseudoelementary atom or literal of $C$ is one whose letter is a (non-)pseudoelementary letter of $C$. 

The names of the following rules have been selected with the conclusion-to-premise view in mind.\vspace{10pt}

\begin{center} {\bf The inference rules of   $\cleighteen$}\end{center}

\begin{description}
\item[Por-Commutativity:]  $X[B\mld A]\leadsto X[A\mld B]$.
\item[Pand-Commutativity:] $X[B\mlc A]\leadsto X[A\mlc B]$. 
\item[Por-Associativity:]  $X[A\mld(B\mld C)]\leadsto  X[(A\mld B)\mld C)]$.
\item[Pand-Associativity:] $X[A\mlc(B\mlc C)]\leadsto  X[(A\mlc B)\mlc C)]$.
\item[Por-Identity:]  $X[A]\leadsto X[A\mld \bot]$.
\item[Pand-Identity:]  $X[A]\leadsto  X[A\mlc \top]$.  
\item[Por-Domination:]  $X[\top]\leadsto X[ A\mld \top]$.
\item[Pand-Domination:]  $X[\bot]\leadsto  X[A\mlc \bot]$.  
\item[Pand-Distribution:]     $X[(A\mld C)\mlc(B \mld C)]\leadsto  X[(A\mlc B)\mld C]$. 
\item[Chand-Distribution:]  $X[(A\mld C)\adc^c (B \mld C)]\leadsto  X[(A\adc^c B)\mld C]$.
\item[Left Cleansing:] $X\bigl[Y[A]\adc^c C\bigr]\leadsto X\bigl[Y[A\adc^c B]\adc^c C\bigr]$.
\item[Right Cleansing:]  $X\bigl[C \adc^c Y[B]\bigr]\leadsto X\bigl[C \adc^c Y[A\adc^c B]\bigr]$. 
\item[Splitting:] $A,B\leadsto  A\adc^c B$,  where $\adc^c$ does not occur in either premise.   
\item[Trivialization:] $X[\top]\leadsto  X[Q^a\mld \neg Q^b]$,  where $Q$ is a pseudoelementary letter of the conclusion.
\item[Quadrilemma:] $X\Bigl[\Bigl(\bigl(A\mlc (C \adc^b D)\bigr) \adc^a \bigl( B\mlc (C \adc^b D)\bigr)\Bigr)\adc^c \Bigl(\bigl((A\adc^a B)\mlc C\bigr)
      \adc^b\bigl( (A\adc^a B)\mlc D\bigr)\Bigr)\Bigr] \leadsto$  $X[(A\adc^a B)\mlc (C \adc^b D)]$,  where $\adc^c$ does not occur in the conclusion. 
\item[Left Choosing:]  $X[A_1,\ldots,A_n]\leadsto  X[A_1 \add^c B_1,\ldots,A_n \add^c B_n]$, \ where $A_1\add^c B_1$, \ldots, $A_n \add^c B_n$ are all the $\add^c$-rooted subcirquents of the conclusion. 
\item[Right Choosing:] $X[B_1,\ldots,B_n]\leadsto  X[A_1 \add^c B_1,\ldots,A_n \add^c B_n]$,  \ where $A_1\add^c B_1$, \ldots, $A_n \add^c B_n$ are all the $\add^c$-rooted subcirquents of the conclusion.   
\item[Matching:] $X[Q^c,\neg Q^d]\leadsto X[P^a,\neg P^b]$, where $P$ is a non-pseudoelementary letter of the conclusion, and $Q$ is a letter  which does not occur in the conclusion.  
\end{description}

We will be using the word ``Commutativity'' as a common name of Por-Commutativity and Pand-Commutativity.  Similarly for other rules.

A {\bf proof} of a cirquent $A$ in $\cleighteen$ is a sequence $C_1,\ldots,C_n$ ($n\geq 1$) of cirquents such that $C_1=\top$, $C_n=A$ and, for each $i\in\{2,\ldots,n\}$, $C_{i}$ follows by one of the inference rules of $\cleighteen$ from some earlier cirquents of the sequence. Alternatively, a proof of $A$ in $\cleighteen$ can be understood as a tree of cirqunts with $A$ at its root and $\top$ at its leaves, where the cirquent of each parent node follows from the cirquents of the children nodes by one of the inference rules of $\cleighteen$. In either case, as we see,  
  $\top$ acts as the (only) {\bf axiom} of \cleighteen.  We agree that, throughout the rest of this article, the default meaning of the expression ``$A$ is {\bf provable}'' is ``$A$ is provable in $\cleighteen$'', i.e., ``there is a proof of $A$  in $\cleighteen$''.

\begin{example}{\em  While the cirquent $ (P^0\mlc P^1)\mld \neg P^2$ is unprovable,  its ``modification'' $(P^1\mlc P^1)\mld \neg P^{2}$ (where $P$ has been made pseudoelementary) can be proven as follows:}

\begin{center} 
\begin{picture}(100,27)
\put(42,14){\small  Axiom}
\put(72,17){\line(1,0){46}}
\put(-6,17){\line(1,0){46}}
\put(53,0){$\top$ }
\end{picture}
\end{center}

\begin{center} 
\begin{picture}(100,17)
\put(27,14){\small  Pand-Identity}
\put(86,17){\line(1,0){32}}
\put(-6,17){\line(1,0){32}}
\put(43,0){$\top\mlc\top$ }
\end{picture}
\end{center}

\begin{center} 
\begin{picture}(100,17)
\put(28,14){\small Trivialization}
\put(86,17){\line(1,0){32}}
\put(-6,17){\line(1,0){32}}
\put(0,0){$(P^{1}\mld \neg P^2)\mlc (P^1\mld \neg P^2)$ }
\end{picture}
\end{center}

\begin{center} 
\begin{picture}(100,27)
\put(20,24){\small Pand-Distribution}
\put(93,27){\line(1,0){25}}
\put(-6,27){\line(1,0){25}}
\put(20,10){$(P^1\mlc P^1)\mld \neg P^{2}$}
\end{picture}
\end{center}

{\em What makes  $(P^1\mlc P^1)\mld \neg P^{2}$ valid is that, semantically, it is virtually the same as  $P^1\mld \neg P^{2}$. $\top$ can easily win the latter (under whatever interpretation) by imitating, in $\neg P^2$,  the moves that $\bot$ makes in $ P^{1}$, and vice versa. This strategy essentially achieves  the effect of letting the adversary play against itself.  A similar strategy in not applicable in the case of $ (P^0\mlc P^1)\mld \neg P^2$ though, as $\neg P^2$  can be copycatted with either $P^0$ or $P^1$, but not both. Such a game can be lost if the adversary makes different moves in $P^0$ and $P^1$.}
\end{example}

\begin{example} {\em Here is a proof of another ``modification'' of the unprovable $ (P^0\mlc P^1)\mld \neg P^2$:}

\begin{center} 
\begin{picture}(130,27)
\put(5,14){\small Axiom\hspace{79pt} Axiom}
\put(34,17){\line(1,0){18}}
\put(-15,17){\line(1,0){18}}
\put(95,17){\line(1,0){18}}
\put(142,17){\line(1,0){18}}
\put(16,0){$\top\hspace{100pt}\top$ }
\end{picture}
\end{center}

\begin{center} 
\begin{picture}(130,17)
\put(-7,14){\small Trivialization\hspace{51pt} Trivialization}
\put(49,17){\line(1,0){8}}
\put(-16,17){\line(1,0){8}}
\put(94,17){\line(1,0){8}}
\put(156,17){\line(1,0){8}}
\put(0,0){$P^{0}\mld \neg P^2\hspace{67pt}P^1\mld \neg P^2$ }
\end{picture}
\end{center}

\begin{center} 
\begin{picture}(130,17)
\put(57,14){\small Splitting}
\put(93,17){\line(1,0){72}}
\put(-16,17){\line(1,0){72}}
\put(16,0){$(P^{0}\mld \neg P^2)\adc^3 (P^1\mld \neg P^2)$ }
\end{picture}
\end{center}

\begin{center} 
\begin{picture}(130,27)
\put(36,24){\small Chand-distribution}
\put(113,27){\line(1,0){52}}
\put(-16,27){\line(1,0){52}}
\put(36,10){$(P^0\adc^3 P^1)\mld \neg P^{2}$}
\end{picture}
\end{center}

\end{example}

\begin{example} \label{pluss}
{\em Below is a proof of a cirquentized version of Blass's \cite{Bla92} principle. In it, Rule$^+$ should be understood as an application of Rule in combination of whatever number of applications of Associativity and Commutativity rules. }

\begin{center} 
\begin{picture}(130,27)
\put(47,14){\small Axiom}
\put(80,17){\line(1,0){212}}
\put(-170,17){\line(1,0){212}}
\put(58,0){$\top$ }
\end{picture}
\end{center}

\begin{center} 
\begin{picture}(130,17)
\put(25,14){\small  Por-Domination$^+$ }
\put(98,17){\line(1,0){193}}
\put(-170,17){\line(1,0){193}}
\put(31,0){$ \top  \mld  (Q^7\mld R^3)$ }
\end{picture}
\end{center}

\begin{center} 
\begin{picture}(130,17)
\put(34,14){\small   Trivialization  }
\put(92,17){\line(1,0){200}}
\put(-170,17){\line(1,0){200}}
\put(11,0){$ ( S^8\mld \neg S^4)  \mld  (Q^7\mld R^3)$  }
\end{picture}
\end{center}

\begin{center} 
\begin{picture}(130,17)
\put(-23,14){\small  Por-Associativity$^+$, Por-Commutativity$^+$ }
\put(150,17){\line(1,0){142}}
\put(-170,17){\line(1,0){142}}
\put(3,0){$ (\neg R^3\mld \neg S^4)  \mld  (Q^7\mld S^8)$ }
\end{picture}
\end{center}

\begin{center} 
\begin{picture}(130,17)
\put(31,14){\small   Pand-Identity$^+$  }
\put(95,17){\line(1,0){196}}
\put(-170,17){\line(1,0){196}}
\put(-17,0){$ \top \mlc\bigl((\neg R^3\mld \neg S^4)  \mld  (Q^7\mld S^8)\bigr)$}
\end{picture}
\end{center}

\begin{center} 
\begin{picture}(130,17)
\put(27,14){\small   Por-Domination  }
\put(95,17){\line(1,0){196}}
\put(-170,17){\line(1,0){196}}
\put(-49,0){$\bigl((P^5\mld \neg S^4)  \mld \top \bigr)\mlc\bigl((\neg R^3\mld \neg S^4)  \mld  (Q^7\mld S^8)\bigr)$ }
\end{picture}
\end{center}

\begin{center} 
\begin{picture}(130,17)
\put(34,14){\small   Trivialization  }
\put(92,17){\line(1,0){200}}
\put(-170,17){\line(1,0){200}}
\put(-69,0){$\bigl((P^5\mld \neg S^4)  \mld (R^6\mld \neg R^3)\bigr)\mlc\bigl((\neg R^3\mld \neg S^4)  \mld  (Q^7\mld S^8)\bigr)$  }
\end{picture}
\end{center}

\begin{center} 
\begin{picture}(130,17)
\put(-23,14){\small  Por-Associativity$^+$, Por-Commutativity$^+$ }
\put(150,17){\line(1,0){142}}
\put(-170,17){\line(1,0){142}}
\put(-69,0){$\bigl((\neg R^3\mld \neg S^4)  \mld (P^5\mld R^6)\bigr)\mlc\bigl((\neg R^3\mld \neg S^4)  \mld(Q^7\mld S^8)\bigr)$  }
\end{picture}
\end{center}

\begin{center} 
\begin{picture}(130,17)
\put(21,14){\small  Pand-Distribution$^+$}
\put(100,17){\line(1,0){190}}
\put(-170,17){\line(1,0){190}}
\put(-35,0){$(\neg R^3\mld \neg S^4)  \mld\bigl((P^5\mld R^6)\mlc(Q^7\mld S^8)\bigr)$  }
\end{picture}
\end{center}

\begin{center} 
\begin{picture}(130,17)
\put(31,14){\small  Pand-Identity$^+$}
\put(95,17){\line(1,0){196}}
\put(-170,17){\line(1,0){196}}
\put(-58,0){ $\top \mlc \Bigl( (\neg R^3\mld \neg S^4)  \mld\bigl((P^5\mld R^6)\mlc(Q^7\mld S^8)\bigr)\Bigr)$ }
\end{picture}
\end{center}

\begin{center} 
\begin{picture}(130,17)
\put(27,14){\small   Por-Domination}
\put(95,17){\line(1,0){196}}
\put(-170,17){\line(1,0){196}}
\put(-88,0){ 
$\bigl((\neg P^1\mld S^8) \mld \top\bigr)\mlc \Bigl( (\neg R^3\mld \neg S^4)  \mld\bigl((P^5\mld R^6)\mlc(Q^7\mld S^8)\bigr)\Bigr)$ }
\end{picture}
\end{center}

\begin{center} 
\begin{picture}(130,17)
\put(34,14){\small   Trivialization}
\put(92,17){\line(1,0){200}}
\put(-170,17){\line(1,0){200}}
\put(-105,0){ 
$\bigl((\neg P^1\mld S^8) \mld(Q^7\mld \neg Q^2 )\bigr)\mlc \Bigl( (\neg R^3\mld \neg S^4)  \mld\bigl((P^5\mld R^6)\mlc(Q^7\mld S^8)\bigr)\Bigr)$ }
\end{picture}
\end{center}

\begin{center} 
\begin{picture}(130,17)
\put(-23,14){\small  Por-Associativity$^+$, Por-Commutativity$^+$}
\put(150,17){\line(1,0){142}}
\put(-170,17){\line(1,0){142}}
\put(-105,0){   $\bigl((\neg P^1\mld\neg Q^2) \mld(Q^7\mld S^8)\bigr)\mlc \Bigl( (\neg R^3\mld \neg S^4)  \mld\bigl((P^5\mld R^6)\mlc(Q^7\mld S^8)\bigr)\Bigr)$  }
\end{picture}
\end{center}

\begin{center} 
\begin{picture}(130,17)
\put(30,14){\small   Pand-Identity$^+$}
\put(94,17){\line(1,0){197}}
\put(-170,17){\line(1,0){197}}
\put(-121,0){ $\Bigl(\top\mlc\bigl((\neg P^1\mld\neg Q^2) \mld(Q^7\mld S^8)\bigr)\Bigr) \mlc \Bigl( (\neg R^3\mld \neg S^4)  \mld\bigl((P^5\mld R^6)\mlc(Q^7\mld S^8)\bigr)\Bigr)$}
\end{picture}
\end{center}

\begin{center} 
\begin{picture}(130,17)
\put(24,14){\small   Por-Domination$^+$}
\put(98,17){\line(1,0){193}}
\put(-170,17){\line(1,0){193}}
\put(-148,0){$\Bigl(\bigl(\top\mld(\neg Q^2\mld R^6)\bigr)\mlc\bigl((\neg P^1\mld\neg Q^2) \mld(Q^7\mld S^8)\bigr)\Bigr) \mlc \Bigl( (\neg R^3\mld \neg S^4)  \mld\bigl((P^5\mld R^6)\mlc(Q^7\mld S^8)\bigr)\Bigr)$ }
\end{picture}
\end{center}

\begin{center} 
\begin{picture}(130,17)
\put(34,14){\small  Trivialization}
\put(92,17){\line(1,0){200}}
\put(-170,17){\line(1,0){200}}
\put(-170,0){$\Bigl(\bigl((P^5\mld\neg P^1)\mld(\neg Q^2\mld R^6)\bigr)\mlc\bigl((\neg P^1\mld\neg Q^2) \mld(Q^7\mld S^8)\bigr)\Bigr) \mlc \Bigl( (\neg R^3\mld \neg S^4)  \mld\bigl((P^5\mld R^6)\mlc(Q^7\mld S^8)\bigr)\Bigr)$  }
\end{picture}
\end{center}

\begin{center} 
\begin{picture}(130,17)
\put(-22,14){\small  Por-Associativity$^+$, Por-Commutativity$^+$}
\put(144,17){\line(1,0){145}}
\put(-170,17){\line(1,0){145}}
\put(-170,0){  $\Bigl(\bigl((\neg P^1\mld\neg Q^2) \mld (P^5\mld R^6)\bigr)\mlc\bigl((\neg P^1\mld\neg Q^2) \mld(Q^7\mld S^8)\bigr)\Bigr) \mlc \Bigl( (\neg R^3\mld \neg S^4)  \mld\bigl((P^5\mld R^6)\mlc(Q^7\mld S^8)\bigr)\Bigr)$ }
\end{picture}
\end{center}

\begin{center} 
\begin{picture}(130,17)
\put(22,14){\small  Pand-Distribution}
\put(100,17){\line(1,0){190}}
\put(-170,17){\line(1,0){190}}
\put(-130,0){ $\Bigl((\neg P^1\mld\neg Q^2) \mld\bigl((P^5\mld R^6)\mlc(Q^7\mld S^8)\bigr)\Bigr) \mlc \Bigl( (\neg R^3\mld \neg S^4)  \mld\bigl((P^5\mld R^6)\mlc(Q^7\mld S^8)\bigr)\Bigr)$ }
\end{picture}
\end{center}

\begin{center} 
\begin{picture}(130,17)
\put(22,14){\small  Pand-Distribution}
\put(100,17){\line(1,0){190}}
\put(-170,17){\line(1,0){190}}
\put(-70,0){ $\bigl((\neg P^1\mld\neg Q^2\bigr)\mlc\bigl(\neg R^3\mld \neg S^4) \bigr)\mld\bigl((P^5\mld R^6)\mlc(Q^7\mld S^8)\bigr)$ }
\end{picture}
\end{center}

\begin{center} 
\begin{picture}(130,27)
\put(20,24){\small Matching (3 times)}
\put(100,27){\line(1,0){190}}
\put(-170,27){\line(1,0){190}}
\put(-70,10){ $\bigl((\neg P^1\mld\neg P^2\bigr)\mlc\bigl(\neg P^3\mld \neg P^4) \bigr)\mld\bigl((P^5\mld P^6)\mlc(P^7\mld P^8)\bigr)$}
\end{picture}
\end{center}

\end{example}

\section{The preservation lemma}
\begin{lemma}\label{pres}

1.  Each application $\vec{A}\leadsto B$ of any of the  rules of $\cleighteen$ preserves validity in the premise-to-conclusion direction, i.e., if all premises from $\vec{A}$ are valid, then so is the conclusion $B$.

2.  Each application $\vec{A}\leadsto B$  of any of the  rules of $\cleighteen$ other than Choosing and Matching also preserves validity in the conclusion-to-premise direction, i.e.,  if $B$ is valid, then so are all premises from $\vec{A}$.

\end{lemma}

\begin{proof}
Since $\cleighteen$ builds upon $\clsixteen$, the parts of the  present proof dealing with rules other than Trivialization and Matching are essentially reproductions  from the proof of a similar lemma for $\clsixteen$ found in \cite{ifcolog}. Of course, this does not extend to the entire proof. $\clsixteen$ simply has no counterpart for Matching. And, while it has a same-name and similar-looking rule of Trivialization, the two versions of this rule are entirely different in their essence due to one dealing with atoms interpreted as elementary games and the other  with atoms standing for arbitrary static games.

As in \cite{ifcolog}, throughout this and some later proofs, when trying to show that a given machine $\cal H$ is a solution of a given game $G$, we implicitly rely on what is called the ``clean environment assumption'' (\cite{iran}). According to it, $\cal H$'s environment never makes moves that are not legal moves of $G$ --- that is, a run is always $\bot$-legal. Assuming that this condition is satisfied is legitimate, because, if $\cal H$'s environment makes the run $\bot$-illegal, then $\cal H$ has nothing to answer for and wins immediately. 

{\em Commutativity}, {\em Associativity}, {\em Identity}, {\em Domination}, {\em Distribution} and {\em Cleansing}.   It is not hard to see that, for whatever interpretation $^*$, if $E\leadsto F$ is an application of any of these rules,  then $E^*$  and $F^*$ are identical as games. This means that a logical solution of $E$ is automatically a logical solution of $F$, and vice versa. The least straightforward case here is Cleansing, so let us just look at this rule.  Consider an application $X\bigl[Y[A]\adc^c C\bigr]\leadsto X\bigl[Y[A\adc^c B]\adc^c C\bigr]$ of (Left) Cleansing, an arbitrary interpretation $^*$ and an arbitrary $^*$-legal run $\Gamma$. We want to show that $\Gamma$ is a $\top$-won run of $E^*$ iff it is a $\top$-won run of $F^*$. If $\adc^c$ is $\Gamma$-unresolved, then $\Gamma$ is a $\top$-won run of both the $Y[A\adc^c B]\adc^c C$ component of the conclusion  and the  $Y[A]\adc^c C$ component of the premise. Since the conclusion and the premise only differ in that one has $Y[A\adc^c B]\adc^c C$ where the other has $Y[A]\adc^c C$, we find that $\Gamma$ is a $\top$-won run of both games or neither. Now assume $\adc^c$ is resolved, i.e., $\Gamma$ contains the labmove $\bot\adc^c\hspace{-2pt}.i$ ($i=0$ or $i=1$). If $i=1$, then $\Gamma$ is a $\top$-won run of  $X[Y[A\adc^c B]\adc^c C]$  iff it is a $\top$-won run of $X[C]$  iff it is a $\top$-won run of $X[Y[A]\adc^c C]$.  And if $i=0$, then $\Gamma$ is a $\top$-won run of 
$ X\bigl[Y[A\adc^c B]\adc^c C\bigr]$ iff it is a $\top$-won run of $ X\bigl[Y[A\adc^c B]\bigr]$ iff it is a $\top$-won run of $ X\bigl[Y[A]\bigr]$ iff it is a $\top$-won run of $ X\bigl[Y[A\adc^c C]\bigr]$. In either case, as we see, $\Gamma$ is a   $\top$-won run of the conclusion  iff it is a   $\top$-won run of the premise.

 {\em Splitting}. Consider an application $A_0,A_1\leadsto  A_0\adc^c A_1$ of this rule. 

{\bf (i)} For the {\em premises-to-conclusion direction}, assume  HPMs ${\cal M}_0$  and ${\cal M}_1$  are logical solutions of $A_0$  and $A_1$, respectively. 
Let $\cal H$ be an HPM that, at the beginning of the play, waits till the environment makes the move $\adc^c\hspace{-1pt}.i$ for $i=0$ or $i=1$. After that, $\cal H$ starts simulating the machine ${\cal M}_i$.  In this process, $\cal H$  simultaneously monitors its own run tape as well the imaginary run tape of ${\cal M}_i$, reading the contents of those two tapes in the left-to-right direction so that it does not process the same labmove twice.   Every time such monitoring detects that ${\cal M}_i$ has made a new (not-yet-seen) move, $\cal H$ makes the same move in the real play; and every time the monitoring detects that $\cal H$'s environment has made a new move, $\cal H$ appends the same ($\bot$-labeled) move to the content of the imaginary run tape of ${\cal M}_i$. 
 In more relaxed and intuitive terms, what we just said about the actions of $\cal H$ can be put as ``$\cal H$ plays exactly like ${\cal M}_i$ would play''.   Later, in similar situations, we may describe and analyze HPMs in relaxed terms, without going into technical details of simulation and without perhaps even using the word ``simulation''. Since we  exclusively deal  with static games  (Fact \ref{stg}), this relaxed approach is safe and valid. Namely, in our present case, a run $\Gamma$ generated by $\cal H$ can easily be seen to be a $\top$-delay of the run $\Upsilon$ generated by ${\cal M}_i$ in the corresponding scenario, which, due to the games being static, is just as good as---or ``even better than''---if $\Gamma$ and $\Upsilon$ were identical:  if $\Upsilon$ is a $\top$-won run of $A_{i}^{*}$, then so is $\Gamma$, which, due to the presence of the labmove $\bot \adc^c\hspace{-1pt}$, makes $\Gamma$ also a $\top$-won run of $(A_0\adc^c\hspace{-1pt} A_1)^*$. But $\Upsilon$ is indeed a $\top$-won run of $A_{i}^{*}$, because ${\cal M}_i$ is a logical solution of $A_i$. Hence, $\Gamma$ is guaranteed to be a solution of  $(A_0\adc^c\hspace{-1pt} A_1)^*$ and hence, as $^*$ was arbitrary, a logical solution of $A_0\adc^c A_1$. 

{\bf (ii)} For the {\em conclusion-to-premises} direction, assume $\cal H$ is a logical solution of $A_0\adc^c\hspace{-1pt} A_1$. Let ${\cal M}_0$ (resp. ${\cal M}_1$) be an HPM that plays just like $\cal H$ would play in the scenario where, at the very start of the play, $\cal H$'s adversary made the move $\adc^c\hspace{-1pt}.0$ (resp. $\adc^c\hspace{-1pt}.1$). Obviously  ${\cal M}_0$ and ${\cal M}_1$   are logical solutions of $A_0$ and $A_1$, respectively.

{\em Trivialization}. Consider an application $X[\top]\leadsto  X[Q^a\mld \neg Q^b]$ of this rule. Let us call moves of the form $Q^a.\alpha$ and $\neg Q^b.\alpha$ (for whatever $\alpha$) {\bf privileged}.  

{\bf (i)} For the premise-to-conclusion direction, assume $\cal M$ is a logical solution of the premise. Let ${\cal H}$ be an HPM that makes exactly the same non-privileged moves and in the same order as $\cal M$ does in the scenario where the imaginary environment of the latter: (i) never makes privileged moves, and (ii)  makes exactly the same non-privileged moves and in the same order as ${\cal H}$'s real environment does. As for privileged moves,  $\cal H$ applies copycat between the $Q^a$ and $\neg Q^b$ components: in parallel with acting as just described,  $\cal H$ keeps scanning its run tape and every time it sees a new labmove $\bot Q^a.\alpha$ 
on its run tape, it makes the move $\neg Q^b.\alpha$, and vice versa: whenever it sees $\bot \neg Q^b.\beta$ on the run tape, it makes the move  
$Q^a.\beta$. Let $\Gamma$ be a run generated by ${\cal H}$ this way, and $\Upsilon$ the run generated by $\cal M$ in the above scenario of $\cal M$'s environment never making priviledged moves but otherwise acting exactly like $\cal H$'s environment acted when $\Gamma$ was generated. 
Consider an arbitrary interpretation $^*$. Let $^+$ be the interpretation such that:
\begin{itemize}
  \item If $\Gamma^{Q^{a}\hspace{-2pt}.}$ is a $\top$-won run of $Q^*$, then $Q^+=\top$, and otherwise $Q^+=\bot$;
  \item $^+$ agrees with $^*$ on all other (other than $Q$) letters.
\end{itemize}
By the clean environment assumption, $\Gamma$ is a $\bot$-legal run of the conclusion under $^*$. With a moment's thought, this can be seen to imply that 
$\Upsilon$ is a $\bot$-legal run of the premise under $^+$. This means that 
\begin{equation}\label{hua}
\mbox{\em $\Upsilon$ is a $\top$-won legal run of the premise under $^+$},
\end{equation}
 because  $\Upsilon$ is a run generated by the logical solution $\cal M$ of the premise. The just-established (in (\ref{hua})) fact that 
$\Upsilon$ is a $^+$-legal run, 
 again with a moment's thought,  can also be seen to guarantee that $\Gamma$ is $^*$-legal. Keeping this in mind, our goal now is to show that $\Gamma$ is a $\top$-won run of the conclusion under $^*$. 

By the definition of $\neg$, either $\Gamma^{Q^a\hspace{-2pt}.}$ is a $\top$-won run of $Q^*$  or 
$\overline{\Gamma^{Q^a\hspace{-1pt}.}}$ is a $\top$-won run of $\neg Q^*$.
But the above-described copycat subroutine obviously guarantees that $\Gamma^{\neg Q^b\hspace{-1pt}.}$ is a $\top$-delay of $\overline{\Gamma^{Q^a\hspace{-1pt}.}}$. Hence, due to the static property of $\neg Q^*$, we have:
\begin{equation}\label{eq5.5}
\mbox{\em Either $\Gamma^{Q^a\hspace{-1pt}.}$ is a $\top$-won run of $Q^*$,  or $\Gamma^{\neg Q^b\hspace{-1pt}.}$ 
is a $\top$-won run of $\neg Q^*$. }\end{equation} 
The above immediately implies that 
\begin{equation}\label{shou}
\mbox{\em $\Gamma$ is a $\top$-won run of $Q^a\mld \neg Q^b$ under $^*$.}\end{equation} 

For any subcirquent $A$ of the conclusion, we define $\underline{A}$  as the result of substituting $\top$ for every occurrence (if any) of $Q^a\mld\neg Q^b$ in $A$. We claim the following: 

\begin{equation}\label{shou2}
\mbox{\em For any subcirquent $A$ of the conclusion, if $\win{\underline{A}^+}{}\seq{\Upsilon}=\top$, then $\win{A^*}{}\seq{\Gamma}=\top$.}\end{equation} 
This claim can be proven by induction on the complexity of $A$. Assume $\win{\underline{A}^+}{}\seq{\Upsilon}=\top$. 
If $A$ is a positive $Q$-lettered literal, then, due to the fact that $Q$ is pseudoelementary, we have $A=Q^a$. 
Hence, by our choice of $^+$, $\win{A^*}{}\seq{\Gamma}=\top$. With (\ref{eq5.5}) in mind, the case of $A$ being a negative $Q$-lettered literal is similar, so we also have $\win{\neg A^*}{}\seq{\Gamma}=\top$. If $A$ is a non-$Q$-lettered literal, then $A^+=A^*$ and, in view of the way $\Gamma$ was generated together with the static property of $A^*$, it is obvious that $\win{A^*}{}\seq{\Gamma}=\top$. If $A$ is 
$Q^a\mld \neg Q^b$, then $\win{A^*}{}\seq{\Gamma}=\top$ is guaranteed by (\ref{shou}). Finally, if $A$ is any 
other non-literal cirquent, then $\win{A^*}{}\seq{\Gamma}=\top$ can easily be seen to follow from the corresponding inductive hypothesis and 
the way $\Gamma$ was generated. This completes our proof of claim (\ref{shou2}). The latter,  in conjunction with (\ref{hua}), immediately implies that $\Gamma$ is a $\top$-won run of the conclusion under $^*$. Since $^*$ was an arbitrary interpretation and $\Gamma$  an arbitrary run generated by $\cal H$, we conclude that $\cal H$ is a logical solution of $X[Q^a\mld \neg Q^b]$. So, as desired, $X[Q^a\mld \neg Q^b]$ is logically valid.

{\bf (ii)}: The conclusion-to-premise direction is much simpler: just observe that, under whatever interpretation $^*$,  if $\Gamma$ is a $\top$-won  run of the conclusion, then it is also (``even more so'') a $\top$-won run of the premise, meaning that any logical solution of the conclusion is also a logical solution of the premise. 

{\em Quadrilemma}. Consider an application of this rule with premise (\ref{prm}) and conclusion (\ref{cnc}):
\begin{equation}\label{prm}
X\Bigl[\Bigl(\bigl(A\mlc (C \adc^b D)\bigr) \adc^a \bigl( B\mlc (C \adc^b D)\bigr)\Bigr)\adc^c \Bigl(\bigl((A\adc^a B)\mlc C\bigr)
      \adc^b\bigl( (A\adc^a B)\mlc D\bigr)\Bigr)\Bigr]
\end{equation}
\begin{equation}\label{cnc}
 X[(A\adc^a B)\mlc (C \adc^b D)].
\end{equation}

{\bf (i)} For the premise-to-conclusion direction,  assume  $\cal M$ is a logical solution of (\ref{prm}). Let $\cal H$ be an HPM that, until  it sees that its environment has made either $\adc^a$ or $\adc^b$ resolved, plays just like $\cal M$ would play in the scenario where $\adc^c$ is not (yet) resolved but otherwise $\cal M$'s imaginary environment is making the same moves as $\cal H$'s real environment is making. If and when it sees that a move $\adc^a\hspace{-1pt}.i$ (resp. $\adc^b\hspace{-1pt}.i$) has been made by its environment, $\cal H$  imagines that $\cal M$'s environment has correspondingly made not only the same move $\adc^a\hspace{-1pt}.i$ (resp. $\adc^b\hspace{-1pt}.i$) but also  $\adc^c\hspace{-1pt}.0$  (resp. $\adc^a\hspace{-1pt}.1$), and  continues playing exactly as $\cal M$ would continue playing in that case. With a little thought, $\cal H$ can be seen to be a logical solution of the conclusion.

{\bf (ii)} For the conclusion-to-premise direction,  observe that, under any interpretation,  any $\top$-won run of (\ref{cnc})   is also (``even more so'') a $\top$-won run of the premise. Hence a logical  solution of the conclusion  is automatically also a  logical solution of the premise. 

{\em Choosing}. Here we only need to show the preservation of logical validity in the premise-to-conclusion direction. Consider an application $X[A_1,\ldots,A_n]\leadsto  X[A_1 \add^c B_1,\ldots,A_n \add^c B_n]$ of Left Choosing, and assume $\cal M$ is a logical 
solution of the premise. Let ${\cal H}$ be an HPM that, at the beginning of the game, makes the move $\add^c\hspace{-2pt}.0$, after which it plays exactly as 
$\cal M$ would. Obviously $\cal H$ 
is a logical solution of the conclusion. Right Choosing will be handled in a similar way. 

{\em Matching}. Again, we only need to consider the premise-to-conclusion direction. Pick an application $X[Q^c,\neg Q^d]\leadsto X[P^a,\neg P^b]$ of this rule. Assume $\cal M$ is a logical solution of the premise. We claim that the same $\cal M$ is a logical solution of the conclusion. Indeed, consider any interpretation $^*$. We want to show that ${\cal M}\models(X[P^a,\neg P^b])^*$. 
Let $^+$ be the interpretation such that $Q^+=P^*$ and $^+$ agrees with $^*$ all other (than $Q$) letters. Notice that, since $X[P^a,\neg P^b]$ does not contain $Q$, we have $(X[Q^c,\neg Q^d])^+=(X[P^a,\neg P^b])^*$. Since $\cal M$ is a logical solution of $X[Q^c,\neg Q^d]$, we have ${\cal M}\models (X[Q^c,\neg Q^d])^+$. Hence also 
${\cal M}\models (X[P^a,\neg P^b])^*$. 
\end{proof}

\begin{remark}\label{rfeb5} {\em Lemma \ref{pres} states the existence of certain solutions. A look back at our proof of the lemma reveals that, in fact, this existence is constructive. Namely, in the case of clause (a) of Lemma \ref{pres}, for any given rule, there is a $^*$-independent effective procedure that extracts an HPM $\cal M$ from the premise(s), the conclusion and HPMs that purportedly solve the premises under $^*$; as long as these purported solutions are indeed solutions, $\cal M$ is a solution of the conclusion under $^*$. Similarly for clause (b). }
\end{remark} 

\section{Soundness and completeness}

We start this section with some terminological conventions. Two literals are said to be {\bf same-lettered} iff both are $P$-lettered for some (the same) letter $P$. Next, two literals are {\bf opposite} iff one is positive, the other is negative, and they are same-lettered.

A {\bf disjunctively normal} subcirquent of a given cirquent $C$ is a subcirquent that is (i) either $\adc$-rooted, (ii) or  of the form $A_1\mld\ldots\mld A_n$ ($n\geq 1$), where each $A_i$ either is a literal or is $\add$-rooted;
 besides, there are no opposite non-pseudoelementary literals  of $C$  among $A_1,\ldots,A_n$.

A {\bf conjunctively normal} cirquent is one  of the form $A_1\mlc\ldots\mlc A_n$ ($n\geq 1$), where each $A_i$  is disjunctively normal, besides, at least one of the $A_i$ is not $\adc$-rooted.

A  {\bf surface subcirquent} of a given cirquent $C$ is a subcirquent $S$ such that $S$ has at least one  occurrence in $C$ that is not in the scope of a choice connective.

We say that a cirquent $C$ is {\bf pure} iff $C$ is either conjunctively normal, or is $\top$, or is $\bot$,  or has the form $A\adc^c B$, where neither $A$ nor $B$ contains  $\adc^c$. 

The following procedure takes a cirquent $D$ and, after a chain of inference rule applications in the conclusion-to-premise direction, returns a pure cirquent. 
\vspace{10pt}

{\em Procedure} {\bf Purification} applied to a cirquent $D$: If none of the following eight steps is applicable (because the corresponding ``if\ldots'' condition is not satisfied), return $D$; otherwise, apply any applicable step, change the value of $D$ to the result of such an application, and repeat the present procedure.\vspace{5pt} 

{\em Step 1}: If $D$ has a surface subcirquent of the form $(A\mlc B)\mld C$ or $C\mld(A\mlc B)$, change that subcirquent to $(A\mld C)\mlc(B\mld C)$ using Pand-Distribution perhaps in combination with Por-Commutativity.  

{\em Step 2}: If $D$ has a surface subcirquent of the form $(A\adc^c B)\mld C$ or $C\mld(A\adc^c B)$, change that subcirquent to $(A\mld C)\adc^c(B\mld C)$ using Chand-Distribution perhaps in combination with Por-Commutativity.

{\em Step 3}: If $D$ has a surface subcirquent of the form $A_1\mld\ldots\mld A_n$  ($n\geq 2$)  and there are opposite pseudoelementary literals of $D$ among $A_1,\ldots, A_n$, change that subcirquent to $\top$ using Trivialization, perhaps in combination with  Por-Commutativity, Por-Associativity and Por-Domination.   

{\em Step 4}: If $D$ has a surface subcirquent of the form $A_1\mld\ldots\mld A_n$ ($n\geq 2$) and $\top$ is among $A_1,\ldots,A_n$,  change that subcirquent to $\top$ using Por-Domination, perhaps in combination with Por-Commutativity and perhaps several times.

{\em Step 5}: If $D$ has a surface subcirquent of the form $A_1\mlc\ldots\mlc A_n$ ($n\geq 2$) and $\bot$ is among $A_1,\ldots,A_n$,  change that subcirquent to $\bot$ using Pand-Domination, perhaps in combination with Pand-Commutativity amd perhaps several times.   

{\em Step 6}: If $D$ has a surface subcirquent of the form $\top \mlc A$,   $A\mlc\top$,  $\bot \mld A$ or $A\mld\bot$, change that subcirquent to $A$ using Pand-Identity or Por-Identity perhaps in combination with Pand-Commutativity or Por-Commutativity.  

{\em Step 7}: If $D$ has a surface subcirquent of the form 
$(A\adc^a B)\mlc (C\adc^b D)$, change that subcirquent to $\Bigl(\bigl(A\mlc(C\adc^j D)\bigr)\adc^a \bigl(B\mlc(C\adc^b D)\bigr)\Bigr)\adc^c \Bigl(\bigl((A\adc^a B)\mlc C\bigr)\adc^b \bigl((A\adc^a B)\mlc D\bigr)\Bigr) $ (where $c$ is a cluster not present in $D_6$) using Quadrilemma.

{\em Step 8}: If $D$ is of the form $A[E\adc^c F]\adc^c B$  (resp. $A\adc^cB[E\adc^c F]$), change it to $A[E]\adc^c B$ (resp. $A\adc^c B[F]$) using Cleansing.  
\vspace{10pt}

It follows from clause 1 of Lemma \ref{terminate} that the above procedure always terminates.   We call the final cirquent resulting from $D$ according to the above procedure the {\bf purification} of $D$.

Below we use the standard notation $^na$ (``tower of $a$'s of height $n$'') for tetration, defined inductively by $^1 a= a$ and $^{n+1}a=a^{(^na)}$.  So, for instance, $^35=5^{5^5}$. In the definition below, in the role of $a$, we need a ``sufficiently large'' number and, for safety, we choose $5$ as such a number, even if a smaller number might also be sufficient for our purposes. 
\begin{definition}\label{rankdef}
The {\bf rank}   of a cirquent $C$ is the number $\overline{C}$ defined as follows:

\begin{enumerate}
  \item $\overline{C}=1$ if $C$ is a literal, or $\top$, or $\bot$;
 \item $\overline{C} =\overline{A} +\overline{B} $ if $C$ is $A\add^c B$ or $A\adc^c B$;
  \item $\overline{C} =5^{\overline{A} +\overline{B} }$ if $C$ is $A\mlc B$;
  \item $\overline{C} $ is $^{\overline{A} +\overline{B} }5$ if $C$ is $A\mld B$.
\end{enumerate}
\end{definition}

\begin{lemma}\label{terminate}  
Each step of the Purification procedure (when applicable):     \begin{enumerate} 
\item  strictly decreases the rank of $D$, and
\item does not increase the number of non-pseudoelementary literals of $D$.  \end{enumerate}  
\end{lemma}

\begin{proof} Clause 2 of the lemma immediately follows from the simple observation that none of the steps of the Purification procedure creates any new literals. The rest of this proof is for clause 1. 

In view of the monotonicity of addition, exponentiation and tetration, first of all observe that the rank function is monotone in the following sense. Consider a cirquent $A$ together with an occurrence $O$ of a subcirquent $B$ in it. Assume $B'$ is a cirquent with $\overline{B'}< \overline{B}$, and $A'$ results from $A$ by replacing the occurrence $O$ of $B$ by $B'$. Then 
$\overline{A'} < \overline{A}$. Also keep in mind that the rank of a cirquent is always at least 1. 

An application of any given stage of the Purification procedure replaces a subcirquent $X$ of $D$ by some cirquent $Y$. In view of the monotonicity of the rank function, in order to show that such a replacement decreases the rank of $D$, it is sufficient to show that $\overline{X}<\overline{Y}$.   

{\em Step 1}: Each iteration of this step replaces, in $D$,   a subcirquent $(A\mlc B)\mld C$ or $C\mld(A\mlc B)$ by  $(A\mld C)\mlc (B\mld C)$. 
The rank of $(A\mlc B)\mld C$ or $C\mld(A\mlc B)$ is
$ ^{[5^{(\overline{A}+\overline{B})}+\overline{C}]}5$, and the rank of  $(A\mld C)\mlc (B\mld C)$ is  $5^{[^{(\overline{A}+\overline{C})}5+^{(\overline{B}+\overline{C})}5]}$. 
We want to show that $5^{[^{(\overline{A}+\overline{C})}5+^{(\overline{B}+\overline{C})}5]}<^{[5^{\overline{A}+\overline{B}}+\overline{C}]}5$.  

We of course have $\overline{A}+\overline{B}+1< 5^{\overline{A}+\overline{B}}$, whence $\overline{A}+\overline{B}+\overline{C}+1< 
5^{\overline{A}+\overline{B}}+\overline{C}$, whence 
\mbox{$^{[\overline{A}+\overline{B}+\overline{C}+1]}5
<^{[5^{\overline{A}+\overline{B}}+\overline{C}]}5$.} We also have

\[5^{[^{(\overline{A}+\overline{C})}5+^{(\overline{B}+\overline{C})}5]}=  5^{[^{(\overline{A}+\overline{C})}5]}\times 5^{[^{(\overline{B}+\overline{C})}5]}=[^{(\overline{A}+\overline{C}+1)}5]\times [^{(\overline{B}+\overline{C}+1)}5] \leq ^{[\overline{A}+\overline{B}+\overline{C}+1]}5.\] 
 Consequently, $5^{[^{(\overline{A}+\overline{C})}5+^{(\overline{B}+\overline{C})}5]}<^{[5^{\overline{A}+\overline{B}}+\overline{C}]}5$, as desired.

{\em Step 2}: Each iteration of this step replaces a subcirquent $(A\adc^c B)\mld C$ or $C\mld(A\adc^c B)$ by  $(A\mld C)\adc^c (B\mld C)$. The rank of $(A\adc^c B)\mld C$ or $C\mld(A\adc^c B)$ is $^{(\overline{A}+\overline{B}+\overline{C})}5$, and the rank of  $(A\mld C)\adc^c (B \mld C)$ is $^{(\overline{A}+\overline{C})}5+^{(\overline{B}+\overline{C})}5$. Of course the latter is smaller than the former.

{\em Steps 3, 4, 5 and 6}: Obvious.

{\em Step 7}:  The rank of $(A \adc^a B)\mlc (E \adc^b F)$ is $5^{[\overline{A}+\overline{B}+\overline{E}+\overline{F}]}$, and   the rank of  
\[\Bigl(\bigl(A\mlc (E \adc^b F)\bigr)\adc^a \bigl(B\mlc (E \adc^b F)\bigr)\Bigr) 
\adc^c \Bigl(\bigl((A\adc^a B)\mlc E\bigr)\adc^b\bigl((A\adc^a B)\mlc D\bigr)\Bigr)\]
is $5^{(\overline{A}+\overline{E}+\overline{F})}+5^{(\overline{B}+\overline{E}+\overline{F})}+5^{(\overline{A}+\overline{B}+\overline{E})}+ 5^{(\overline{A}+\overline{B}+\overline{F})}$. Obviously the latter is smaller than the former.

{\em Step 8}: Obvious.
 \end{proof}

By the {\bf extended rank} of a cirquent $C$ we shall mean the pair $(x,y)$, where $x$ is the rank of $C$ and $y$ is the total number of non-pseudoelementary literals of $C$. We define the linear order $\prec$ (read ``smaller than'', or ``lower than'')  on extended ranks by stipulating that $(x,y)\prec (z,t)$ iff (i) either $x<z$ (ii) or $x=z$ and $y<t$.  Of course, $\prec$  does not create  infinite descending chains, which allows us to run induction on it. 

\begin{lemma}\label{pl}
Let $A$ be an arbitrary cirquent and $B$ be its purification. Then:

1. If $B$ is provable, then so is $A$. 

2. $A$ is valid iff so is $B$.

3. The extended rank of $B$ does not exceed that of $A$. 

4. $B$ is pure. 
\end{lemma}

\begin{proof} {\em Clause 1}: Immediate as each step of the Purification procedure applies one of the inference rules of $\cleighteen$ in the conclusion-to-premise direction. 

{\em Clause 2}: In view of the  reason pointed out above, validity of $A$ imples validity of $B$ by clause 1 of Lemma \ref{pres}. The converse implication follows  from clause 2 of Lemma \ref{pres} and the fact that, when obtaining $B$ from $A$ using the Purification procedure, the rules of Choosing and Matching were not used.

{\em Clause 3}: Immediate from Lemma \ref{terminate}

{\em Clause 4}: Can be verified by observing that, if a cirquent $D$ is not pure, then, depending on the reason, at least one of the (rank-decreasing) steps of the Purification procedure applies, so, sooner or later, all possible reasons for not being pure will be eliminated. Details of this somewhat tedious but otherwise straightforward analysis are left to the reader. 
\end{proof}

\begin{theorem} \label{theo}
A cirquent is valid if (soundness) and only if (completeness) it is provable.
\end{theorem}

\begin{proof} The soundness part is immediate from Lemma \ref{pres} and the obvious fact that the axiom $\top$ is valid. The rest of this section is devoted to a proof of the completeness part. 
Pick an arbitrary cirquent $A$ and assume it is valid. Our goal is to show that $A$ is provable, which will be done by induction on the extended rank of $A$.  Let $B$ be the purification of  $A$. According to clauses 2-4 of Lemma \ref{pl}, $B$ is a valid, pure  cirquent whose extended rank does not exceed that of $A$. Also, by clause 1 of the same lemma, if $B$ is provable, then so is $A$. Hence, in order to show that $A$ is provable, it will suffice to show---as we are going to do---that $B$ is provable. In view of $B$'s being pure, the following cases cover all possibilities for $B$. 

One case is that $B=\bot$. But this is  ruled out because then,   of course, $B$ would not be valid. 

Another case is that $B=\top$. But then $B$ is an axiom and hence provable.

Yet another case is that $B$ has the form $E_0\adc^c E_1$, where neither $E_0$ nor $E_1$ contain the connective $\adc^c$. By clause 2 of Lemma \ref{pres}, both $E_0$ and $E_1$ are valid, because $B$ follows from them by Splitting.  Of course, the (extended) rank of either cirquent is smaller than that of $B$. Hence, by the induction hypothesis, both $E_0$ and $E_1$ are provable. But then so is $B$, as it follows from $E_0$ and $E_1$ by Splitting.

The final possibility for $B$, to which the rest of this proof is devoted, is that it is conjunctively normal. This means $B$ has the form $C_1\mlc \ldots\mlc C_m$ ($n\geq 1$),  and at least one conjunct of it---without loss of generality assume it is $C_1$---is of the form $E_1\mld\ldots\mld E_n$ ($n\geq 1$), where each $E_i$ is either a literal or a  $\add$-rooted cirquent; besides, there is no pair of opposite pseudoelementary literals of $B$  among $E_1,\ldots,E_n$.  Assume  $\cal M$ is a logical solution of $B$. Notice that then $\cal M$ is automatically also a logical solution of $C_1$. 
Let $L_1,\ldots,L_k$ be all literals that occur in $B$. 
Consider the work of $\cal M$ in the (partial)  scenario where its environment makes the moves  $L_1. 1,\ldots,L_k.k$  right at the beginning of the play, and makes no further moves until $\cal M$ also moves. We claim that, in this scenario, sooner or later, (at least) one of the following events will have taken place:

\begin{description}
\item{Event 1:} There is a cirquent $A_0\add^c A_1$ among  $E_1,\ldots ,E_n$  such that  $\cal M$ has made the move $\add^c\hspace{-1pt}.i$, where $i=0$ or $i=1$.  
\item{Event 2:} There are opposite literals $L_x,L_y$ (where $x,y\in\{1,\ldots,k\}$) among $E_1,\ldots ,E_n$  such that $\cal M$ has made  the two moves   $ L_{y}.x$ and  $L_{x}.y$. That is, in intuitive terms, $\cal M$ has made the same move $x$ in $L_{y}$ as its adversary has made in $L_{x}$, and $\cal M$ has made the same move $y$ in $L_x$ as its adversary has made in $L_y$.
\end{description}
Indeed, if neither Event ever occurs, then, where $\Upsilon$ is the run generated by $\cal M$ in the above scenario,  it ($\Upsilon$) will be a $\top$-lost run of $C_1$ under an interpretation $^*$ that interprets the letter of every  literal $L\in\{E_1,\ldots,E_n\}$ as a game such that $\Upsilon^{L.}$ is a $\top$-lost run of $L^*$; this, however, is impossible in view of our assumption that $\cal M$ is a logical solution of $C_1$. 

Our further reasoning on the path to showing that $B$ is provable depends on which of the above two Events occurs (if both occur, and/or occur multiple times, just pick one arbitrarily).

{\em Case of Event 1.} We construct a new machine $\cal N$ as follows. 
 Throughout the play, $\cal N$ maintains a list of what we call {\bf couples}, with each couple being an unordered pair $\{R, S\}$  of opposite literals of $B$. When $\{R, S\}$ is on this list, we say that $R$ and $ S$  are {\bf spouses} of each other. Whenever a literal has a spouse, it is said to be {\bf married}. Initially the list of couples is empty. During the play, a new couple can be added to the list of couples and, once a couple is on the list, it remains on it forever.  
$\cal N$ simulates   $M$ in the scenario where the adversary of the latter made the earlier mentioned moves $ L_1. 1,\ldots, L_k.k$  right at the beginning of 
the play. During this process, $\cal N$ simultaneously reads the labmoves on its run tape and the imaginary run tape of $\cal M$, and does this in the left-to-right fashion so that it never reads the same labmove twice.   If on $\cal M$'s run tape  $\cal N$ sees  two  labmoves $\top  L_u.v$ and $\top  L_v.u$ where $u,v\in\{1,\ldots,k\}$ and the literals $L_u,L_v$ are opposite, $\cal N$ declares $\{ L_u, L_v\}$ to be a (new) couple. Also, whenever on its own run tape $\cal N$ sees a new (not-yet-read)  labmove $\bot R.\alpha$ for some married literal $R$, $\cal N$  makes the move $S.\alpha$, where $S$ is the spouse of $R $. 
Other than that, $\cal N$ fully correlates the
 choice moves between the real and imaginary plays. Namely, whenever $\cal N$  sees a new  choice  labmove $\bot \ada^c.i$ on its run tape, it appends the same labmove $\bot \ada^c.i$ to the content of the imaginary run tape of the simulated $\cal M$; and whenever $\cal N$ sees a new   labmove $\top  \add^d.i$ on the imaginary run tape of $\cal M$, $\cal N$ makes the  
move $\add^d.i$ in its real play. 

For runs $\Omega,\Upsilon$, we shall write $\Omega\approx\Upsilon$ iff  one run is a (not necessarily proper) permutation of the other. 

Fix   $\Gamma$  as an arbitrary run generated by $\cal N$, and $\Delta$ as the run incrementally spelled on the imaginary (by $\cal N$) run tape of the simulated $\cal M$  in the process in which  $\Gamma$ was generated.    An analysis of that process, details of which are left to the reader, easily convinces us that:
\begin{equation}
\label{jan21}\mbox{\em for any two $ i, j\in\{ 1,\ldots, k\}$, if $\Delta^{L_{ i}.}\approx\overline{\Delta^{L_ {j}.}}$, then  $\Gamma^{L_{ i}.}$  is a $\top$-delay of $\overline{\Gamma^{L_{ j}.}}$}
\end{equation}

We now claim that 
\begin{equation}\label{jan13}\mbox{\em $\cal N$ is a logical solution of $B$.}\end{equation}
To prove the above claim, deny it for a contradiction. Let  $^*$ be an interpretation such that $\cal N$ does not always win $B^*$. Namely, we can assume that the above-fixed run $\Gamma$ is a $\top$-lost run of  $B^*$. 
We  define a new interpretation $^\circ$ by stipulating that, for every letter $E$, we have:
\begin{description} 
\item[$\legal{E^\circ}{}$] A nonempty run is a legal run of $E^\circ$ iff it has one of the forms $\seq{\bot m}$,  $\seq{\top n}$,  $\seq{\bot m,\top n}$, or $\seq{\top n,\bot m}$, where $n,m\in\{0,1,2,\ldots\}$. 
\item[$\win{E^\circ}{}$] For every $m,n\in\{0,1,2,\ldots\}$, we have:
\begin{enumerate}
\item  $\win{E^\circ}{}\seq{ }=\bot$ (this choice is arbitrary, could have been chosen $\top$ instead).
\item $\win{E^\circ}{}\seq{\bot m}=\bot$.
\item $\win{E^\circ}{}\seq{\top m}=\top$.
\item $\win{E^\circ}{}\seq{\bot m,\top n}=\win{E^\circ}{}\seq{\top n,\bot m}=\wp$, where the player $\wp$ is determined as follows:
 \begin{enumerate}
\item If $\seq{\bot m,\top n}=\Delta^{L.}$ for some positive $E$-lettered  literal  $L$, then $\wp=\win{E^*}{}\seq{\Gamma^{ L.}}$. 
\item Otherwise $\wp=\top$.
\end{enumerate}
\end{enumerate}
\end{description}

From the clean environment assumption for $\cal N$, the way $\Delta$ was generated, and the assumption that $\cal M$ is a logical solution of $B$, with a little thought, it can be seen that 
\begin{equation}\label{little}
\mbox{\em $\Delta$ is a legal, $\top$-won run of $B^\circ$.}
\end{equation}

We now claim that
\begin{equation}\label{jan27}\mbox{\em  For every literal $L$ of $B$,  if $\win{L^*}{}\langle\Gamma^{ L.}\rangle=\bot$, then $\win{L^\circ}{}\langle\Delta^{ L.}\rangle=\bot$.}
\end{equation}

To justify (\ref{jan27}), assume $\win{L^*}{}\langle\Gamma^{ L.}\rangle=\bot$. Our goal is to show that  $\win{L^\circ}{}\langle\Delta^{ L.}\rangle=\bot$.  Let $E$ be the letter of $L$. In view of (\ref{little}) and the way $\Delta$ was generated,  $\Delta^{L.}$ has to be either $\seq{\bot m}$ or $\seq{\bot m,\top n}$ for some $m,n$. 

First, consider the case of $\Delta^{L.}=\seq{\bot m}$. If the literal $L$ is positive, then $L^\circ=E^\circ$, and  the consequent of   (\ref{jan27}) is immediate from clause 2 of our definition of $\win{E^\circ}{}$. And if $L$ is negative, then $L^\circ=\neg E^\circ$, so  $\win{L^\circ}{}\langle\Delta^{ L.}\rangle= \win{L^\circ}{}\langle\bot m\rangle= \win{\neg E^\circ}{}\langle\bot m\rangle=\overline{ \win{E^\circ}{}\langle\top m\rangle}$. But by clause 3 of our definition of $\win{E^\circ}{}$, 
$\win{E^\circ}{}\langle\top m\rangle=\top$, so $\overline{ \win{E^\circ}{}\langle\top m\rangle}=\overline{\top}=\bot$. We thus find again that  the consequent of  (\ref{jan27})  is true. 

Now consider the case of $\Delta^{L.}=\seq{\bot m,\top n}$. If the literal 
$L$ is positive, then  the consequent of   (\ref{jan27}) is immediate from  clause 4(a) of our definition of $\win{E^\circ}{}$. Suppose now $L$ is  negative. 
First, assume there is a positive $E$-lettered literal $K$ such that  $\Delta^{K.}=\seq{\bot n,\top m}$. We thus have $\Delta^{K.}\approx\overline{\Delta^{L.}}$. Our assumption $\win{L^*}{}\langle\Gamma^{ L.}\rangle=\bot$, i.e.,  $\win{\neg E^*}{}\langle\Gamma^{ L.}\rangle=\bot$,  implies that $\win{E^*}{}\overline{\langle\Gamma^{ L.}\rangle}=\top$ and hence, in view of (\ref{jan21}), $\win{E^*}{}\langle\Gamma^{ K.}\rangle=\top$. Consequently, by clause 4(a) of our definition of $\win{E^\circ}{}$,  $\win{E^\circ}{}\seq{\top m,\bot n}=\top$. But this, in turn, implies that  $\win{\neg E^\circ}{}\seq{\bot m,\top n}=\bot$, i.e., $\win{L^\circ}{}\langle\Delta^{ L.}\rangle=\bot$, as desired. Now assume there is no such $K$. Then, by clause 4(b) of our definition of $\win{E^\circ}{}$, we again have $\win{E^\circ}{}\seq{\top m,\bot n}=\top$, which, again, implies the desired $\win{L^\circ}{}\langle\Delta^{ L.}\rangle=\bot$. This completes our justification of (\ref{jan27}).

 By induction on the complexity of $B$, claim (\ref{jan27})  (induction basis) and the fact that $\Gamma$ and $\Delta$ contain exactly the same choice moves, imply that  if $\win{B^*}{}\langle\Gamma\rangle=\bot$, then $ \win{B^\circ}{}\langle\Delta\rangle=\bot$. Hence $\Delta$ is a $\top$-lost run of $B^\circ$, because, by our assumption, $\Gamma$ is a $\top$-lost run of $B^*$. 
 But this is in contradiction with (\ref{little}), and  (\ref{jan13}) is thus proven. 
 
 Let  now $F$ be the result of changing, in  $B$, all subcirquents of the form $X_0\add^c X_1$ to $X_i$, where (if we have already forgotten) $c,i$ are as in the description of Event 1. According to our construction of $\cal N$, $\Gamma$ contains the labmove $\top c.i$ because $\Delta$ does. With this observation in mind, it can be seen that $\cal N$ is a logical solution of (not only $B$ but also) $ F$. Thus $F$ is valid.  Of course, the (extended) rank of $F$ is smaller than that of $B$. Hence, by the induction hypothesis, $F$ is provable. But then so is $B$, as it follows from   $F$ by Choosing.

{\em Case of Event 2.}  Let  $x$, $y$, $L_x$, $L_y$ be as in the  description of Event 2.  We may assume that $L_x=P^a$ snd $L_y=\neg P^b$ for some letter $P$ and clusters $a,b$. Note that, since $B$ is pure, $P$ is not a quasielementary letter of $B$. 
 We pick an arbitrary letter $Q$ not occurring in $B$, and arbitrary clusters $c,d$.   Let $\cal H$ be a machine whose work can be described literally the same way as we earlier described the work of the machine $\cal N$, with the only difference that $\cal H$ declares $\{Q^c,\neg Q^d\}$  to be  a couple instead of $\{ L_x, L_y\}$ . That is, $\cal H$ treats   $Q^c$ and $\neg Q^d$ as $\cal N$ would treat $P^a$ and $\neg P^ b$, respectively.  Let now $G$ be 
the result of replacing, in $B$, the literals $P^a$ and $\neg P^ b$ by $Q^c$ and $\neg Q^d$, respectively. With $\cal H$ and $ G$ in the roles of $\cal N$ and $B$, respectively, arguing exactly as we did when earlier proving (\ref{jan13}), we find that $\cal H$ is a logical solution of $G$. The extended rank of $G$ is smaller than that of $B$ and therefore, by the induction hypothesis, we can assume that $G$ is provable. But then so is $B$ because it follows from $G$ by Matching. This completes our proof of Theorem \ref{theo}.
\end{proof}

\section{From general-base to mixed-base formalism and back}\label{from} 

As briefly discussed in Section \ref{intr}, technically our system  $\cleighteen$ qualifies as  what CoL calls {\em general-base}, as the letters of its literals are allowed to be interpreted as any static games. Let us refer to such letters and corresponding atoms and literals as {\bf general}.  The language of $\cleighteen$ can be enriched by adding to  it a second, {\bf elementary} sort of letters, which are only allowed to be interpreted as elementary games, i.e., games with no legal moves (meaning that $\emptyrun$ is the only legal run).  This way we would get a what CoL calls {\em mixed-base} formalism. However,  as long as logical validity is concerned, there is no technical need for studying systems in the mixed-base formalism: the latter, as will be shown in this section,  can be easily embedded into the formalism of $\cleighteen$ by ``translating'' elementary letters  into pseudoelementary ones.  

The reason for clustering (general) atoms was to indicate whether different occurrences of same-letter atoms can be realized as different runs of the same game or the same run of the same game. Elementary games, however, do not have different possible (legal) runs. For this reason, there are no reasons for also clustering elementary atoms. 

So, let us consider the extension of the formalism defined in Section \ref{Syntax} obtained by adding to it elementary letters/atoms with no clusters attached, and for which we shall use the lowercase $p,q,\ldots$ as metavariables. With this modification of the concept of an atom in mind, a {\bf mixed-base cirquent} is defined verbatimly the same way as (just) cirquent in Definition \ref{cirquent}. The same goes for the concepts defined in Definitions \ref{clsem} and \ref{krisha}.   The concept of an interpretation $^*$ is modified by requiring that, for every elementary letter $p$, $p^*$ is an elementary game. 

Given two mixed-base cirquents $A$ and $B$, we say that $B$ is a {\bf one-step deelementarization} of $A$ iff, for some elementary letter $p$ occurring in $A$ and some general letter $P$ not occurring in $A$, $B$ is the result of replacing, in $A$, every occurrence (if there are any) of the literal $p$ by $P^0$ and every occurrence (if there are any) of the literal $\neg p$ by $\neg P^0$.\footnote{For our purposes, it does not really matter what the clusters of $P$ and $\neg P$ are or whether they are the same or not.} We say that $B$ is a {\bf full deelementarization} of $A$ iff $B$ contains no elementary atoms and there are cirquents $C_1,\ldots,C_n$ ($n\geq 1$) such that $C_1=A$, $C_n=B$ and, for each $1\leq i < n$, $C_{i+1}$ is a one-step deelementarization of $C_i$. 
\begin{lemma}\label{apr21}
Assume $A$ is a mixed-base cirquent and $B$ is a one-step deelementarization of $A$. Then $A$ is logically valid iff so is $B$. 
\end{lemma}
\begin{proof} Assume $A$ and $B$ are as in the condition of the theorem and, correspondingly, $p$ and $P$ be as in our definition of one-step deelementarization. 

Assume $B$ is valid. Let $\cal M$ be an HPM that is a logical solution of $B$. We claim that the same $\cal M$ is also a logical solution of $A$ (and hence $A$ is valid). To see this, consider an arbitrary interpretation $^*$. Let $^\circ$ be the function such that $P^{\circ}=p^*$ and $^{\circ}$ agrees with $^*$ on all other letters. Having no nonempty legal runs, elementary games are automatically static. This means that $^{\circ}$ is a legitimate interpretation for $B$.  Since $\cal M$ is a logical solution of $B$, we have ${\cal M}\models B^\circ$. But notice that 
  $B^{\circ}$ and  $A^{*}$ are identical as games. So, ${\cal M}\models A^*$. As $^*$ was an arbitrary interpretation, we find that  $\cal M$ is  a logical solution of $A$.

For the opposite direction, assume that $A$ is valid. Let $\cal M$ be an HPM that is a logical solution of $A$. As $A$ does not contain the letter $P$, we may assume that $\cal M$ never makes (the dummy) moves of the form $P^0.\alpha$ or $\neg P^0.\alpha$. Let then $\cal N$ be an HPM that, simulating $\cal  M$, plays as follows. Every time $\cal N$ sees that the simulated $\cal M$ made a move $\alpha$, $\cal N$ makes the same move $\alpha$ in its real play. And every time $\cal N$ sees that its (real) environment made a move $\alpha$, except when $\alpha$ is of the form $P^0.\beta$ or $\neg P^0.\beta$, $\cal N$ appends the labmove $\bot \alpha$ to the imaginary run tape of $\cal M$.  Finally, if $\cal N$ sees that its adversary made a move $P^0.\beta$ (resp. $\neg P^0.\beta$), it makes the copicat move $\gneg P^0.\beta$ (resp. $P^0.\beta$) in the real play. It is left as an easy exercise for the reader to verify that  such an $\cal N$ is a logical solution of $B$ and hence $B$ is valid.  
\end{proof}

As an immediate corollary of Lemma \ref{apr21}, we have:
\begin{theorem}\label{hgh}
A mixed-base cirquent is valid iff so is its (general-base) full deelementarization.
\end{theorem}


\begin{thebibliography}{99}

\bibitem{Avron} A. Avron. {\em A constructive analysis of RM}. {\bf Journal of Symbolic Logic} 52 (1987), No.4, pp. 939-951.

\bibitem{BauerTOCL} M. Bauer. {\em A PSPACE-complete first order fragment of computability logic}.
{\bf ACM Transactions on Computational Logic} 15 (2014), No 1, Paper 1.

\bibitem{BauerLMCS} M. Bauer. {\em The computational complexity of propositional cirquent calculus}.
{\bf Logical Methods is Computer Science} 11 (2015), Issue 1, Paper 12, pp. 1-16.

\bibitem{Bla72} A. Blass. {\em Degrees of indeterminacy of games}. {\bf Fundamenta Mathematicae} 77 (1972) 151-166.

\bibitem{Bla92} A. Blass. {\em A game semantics for linear logic}. {\bf Annals of Pure and Applied Logic} 56 (1992), pp. 183-220.

\bibitem{Brunnler01} K. Brunnler and A. Tiu. {\em A local system for classical logic}. In: {\bf Lecture Notes in Computer Science} 2250 (2001), pp. 347-361.


\bibitem{Gug07} P. Bruscoli and A. Guglielmi. {\em On the proof complexity of deep inference}. {\bf ACM Transactions on Computational Logic} 10 (2009), pp. 1-34.

\bibitem{Anupam} A. Das and L. Strassburger. {\em On linear rewriting systems for Boolean logic and some applications to proof theory}.  {\bf Logical Methods in Computer Science} 12 (2016), pp. 1-27.

\bibitem{Gir87}
J.Y. Girard. {\em Linear logic}. {\bf Theoretical computer science} 50 (1887), pp. 1-102.

\bibitem{Guglielmi01}  A. Guglielmi and L. Strassburger. {\em Non-commutativity and MELL in the calculus of structures}. In: {\bf Lecture Notes in Computer Science}  2142 (2001), pp. 54-68. 



\bibitem{Hintikka73} J. Hintikka. {\bf Logic, Language-Games and Information: Kantian Themes in the Philosophy of Logic}. Clarendon Press 1973. 

\bibitem{HS97} J. Hintikka and G. Sandu. {\em Game-theoretical semantics}. In: {\bf Handbook of Logic and Language}. J. van Benthem and A ter Meulen, eds. North-Holland 1997, pp. 361-410. 


\bibitem{Jap03} G. Japaridze. {\em Introduction to computability logic}. {\bf Annals of Pure and Applied Logic} 123 (2003), pp. 1-99.

\bibitem{Japtocl1} G. Japaridze. {\em Propositional computability logic I}. {\bf ACM Transactions on Computational Logic} 7 (2006), pp. 302-330.

\bibitem{Japtocl2} G. Japaridze. {\em Propositional computability logic II}. {\bf ACM Transactions on Computational Logic} 7 (2006), 
pp.  331-362.

\bibitem{Cirq} G. Japaridze. {\em Introduction to cirquent calculus and abstract resource semantics}. {\bf Journal of Logic and Computation} 16 (2006), pp. 489-532.

   
\bibitem{Japtcs} G. Japaridze. {\em From truth to computability I}. {\bf Theoretical Computer Science} 357 (2006), pp. 100-135.

\bibitem{Japtcs2} G. Japaridze. {\em From truth to computability II}. {\bf Theoretical Computer Science} 379 (2007), pp. 20-52.

\bibitem{Japjsl} G. Japaridze. {\em The logic of interactive Turing reduction}. {\bf Journal of Symbolic Logic} 72 (2007), pp. 243-276. 


\bibitem{Propint} G. Japaridze. {\em The intuitionistic fragment of computability logic at the propositional level}. {\bf Annals of Pure and Applied Logic} 147 (2007),  pp. 187-227. 

\bibitem{Japdeep} G. Japaridze. {\em Cirquent calculus deepened}. {\bf Journal of Logic and Computation}  18 (2008),  pp. 983-1028.

\bibitem{Japseq} G. Japaridze. {\em Sequential operators in computability logic}. {\bf Information and Computation} 206 (2008),  pp. 1443-1475. 

\bibitem{Japfour} G. Japaridze. {\em Many concepts and two logics of algorithmic reduction}. {\bf Studia Logica} 91 (2009), pp. 1-24.  

\bibitem{Japfin} G. Japaridze. {\em In the beginning was game semantics}. In: {\bf Games: Unifying Logic, Language, and Philosophy}. O. Majer,
A.-V. Pietarinen and T. Tulenheimo, eds. Springer   2009,  pp. 249-350.
 

\bibitem{Japtoggling} G. Japaridze. {\em Toggling operators in computability logic}. {\bf Theoretical Computer Science} 412 (2011), pp. 971-1004. 

\bibitem{lmcs} G. Japaridze.
{\em From formulas to cirquents in computability logic}. 
{\bf Logical Methods is Computer Science} 7 (2011),  Issue 2 , Paper 1, pp. 1-55.   


\bibitem{separating} G. Japaridze.
{\em Separating the basic logics of the basic recurrences}. {\bf Annals of Pure and Applied Logic} 163 (2012), pp. 377-389.


\bibitem{lbcs} G. Japaridze. {\em A logical basis for constructive systems}. {\bf Journal of Logic and Computation} 22 (2012), pp. 605-642.

\bibitem{taming1} G. Japaridze. {\em The taming of recurrences in computability logic through cirquent calculus, Part I}. {\bf Archive for Mathematical Logic} 52 (2013),  pp. 173-212.
 
\bibitem{taming2} G. Japaridze. {\em The taming of recurrences in computability logic through cirquent calculus, Part II}. {\bf Archive for Mathematical Logic} 52 (2013),  pp. 213-259.

\bibitem{cl12} G. Japaridze. {\em On the system CL12 of computability logic}. 
{\bf Logical Methods is Computer Science} 11 (2015), Issue 3, paper 1, pp. 1-71.  

\bibitem{ifcolog} G. Japaridze. {\em Elementary-base cirquent calculus I: Parallel and choice connectives}. 
{\bf Journal of Applied Logics - IfCoLoG Journal of Logics and their Applications} 5 (2018), no.1, pp. 367-388. 


\bibitem{igpl21}G.Japaridze. 
{\em Elementary-base cirquent calculus II: Choice quantifiers}.
{\bf Logic Journal of the IGPL} 29 (2021), pp. 769-782.

\bibitem{iran} G. Japaridze.
{\em Fundamentals of computability logic}. In: {\bf Selected topics from Contemporary Logic}.  M. Fitting, ed. College Publications, 2021, pp. 477-537. Also published in {\bf Journal of Applied Logics - IfCoLoG Journal of Logics and their Applications} 7 (2020), pp. 1115-1177.

\bibitem{Lor61} P. Lorenzen. {\em Ein dialogisches Konstruktivit\"{a}tskriterium}. In: {\bf Infinitistic Methods}. In: PWN,  Proc. Symp. Foundations of Mathematics, Warsaw, 1961, pp. 193-200.
  

\bibitem{Ver} I. Mezhirov and N. Vereshchagin. {\em On abstract resource semantics and computability logic}. {\bf Journal of Computer and System Sciences} 76 (2010), pp. 356-372. 



\bibitem{qu} M. Qu, J. Luan, D. Zhu and M. Du. {\em On the toggling-branching recurrence of computability logic}. {\bf Journal of Computer Science and Technology} 28 (2013), pp. 278-284. 

\bibitem{XuIGPL} W. Xu and S. Liu. {\em Soundness and completeness of the cirquent calculus system CL6 for computability logic}. {\bf Logic Journal of the IGPL} 20 (2012), pp. 317-330. 

\bibitem{Xure1} W. Xu    and S. Liu. {\em The countable versus uncountable branching recurrences in computability logic}. {\bf Journal of Applied Logic} 10 (2012),  pp. 431-446.

\bibitem{Xure2} W. Xu and S. Liu. {\em The parallel versus branching recurrences in computability logic}. {\bf Notre Dame Journal of Formal Logic} 54 (2013), pp.   61-78. 

\bibitem{XuIf} W. Xu. {\em A propositional system induced by Japaridze's approach to IF logic}.
{\bf Logic Journal of the IGPL}  22 (2014), pp. 982-991.

\bibitem{XuLast} W. Xu. {\em A cirquent calculus system with clustering and ranking}.
{\bf Journal of Applied Logic}  16 (2016), pp. 37-49.


\end{thebibliography}
\end{document}